\documentclass{article}
\usepackage{comment}
\includecomment{arxiv}

\includecomment{nonblinded}

\begin{arxiv}
  \includecomment{articlestyle}
  \usepackage{vmargin}
  \setpapersize{A4}
  \setmargrb{3cm}{2cm}{3cm}{2cm}
\end{arxiv}

\usepackage{natbib}
\usepackage{graphicx}

\usepackage{amsmath,amsfonts}
\DeclareMathOperator{\Prob}{P}
\DeclareMathOperator{\E}{E}

\DeclareMathOperator{\logit}{logit}
\DeclareMathOperator{\ARL}{ARL}
\DeclareMathOperator{\hit}{hit}
\newcommand{\N}{\mathbb{N}}
\newcommand{\R}{\mathbb{R}}
\newcommand{\condweakconv}[1]{\displaystyle{\mathop{\leadsto}^{P
    }_W}}

\newcommand{\xx}[1]{}
\newcommand{\cc}[1]{}
\begin{arxiv}
  \usepackage{amsmath,amsthm}
\end{arxiv}
\newtheorem{remark}{Remark}
\newtheorem{definition}{Definition}
\newtheorem{lemma}{Lemma}
\newtheorem{theorem}{Theorem}
\newtheorem{algorithm}{Algorithm}

\begin{articlestyle}
\title{Guaranteed Conditional Performance of Control Charts via Bootstrap Methods}
\author{}
\begin{nonblinded}
    \author{\normalsize
      Axel Gandy\\
      \normalsize Departments of Mathematics, Imperial College London\\[0.5cm]
      \normalsize Jan Terje Kval\o y\\
      \normalsize Department of Mathematics and Natural Sciences,\\
      \normalsize University of Stavanger, Norway
    }
\end{nonblinded}
\date{}
\end{articlestyle}

\usepackage{xcolor}

\begin{document}

\begin{arxiv}
 \maketitle
\end{arxiv}

\begin{abstract}
  To use control charts in practice, the in-control state usually has
  to be estimated. This estimation has a  detrimental effect on the
  performance of control charts, which is often measured for example
  by the false alarm probability or the average run length.  We
  suggest an adjustment of the monitoring schemes to overcome these
  problems.  It guarantees, with a certain probability, a conditional
  performance given the estimated in-control state.  The suggested
  method is based on bootstrapping the data used to estimate the
  in-control state. The method applies to different types of control
  charts, and also works with charts based on regression models,
  survival models, etc.  If a nonparametric bootstrap is used, the
  method is robust to model errors. We show large sample properties of the
  adjustment.  The usefulness of our approach is demonstrated through
  simulation studies.
\end{abstract}

\begin{arxiv}
\end{arxiv}

\section{Introduction}
Control charts such as the Shewhart chart \citep{Shewhart1931ECo} and
the cumulative sum (CUSUM) chart \citep{Page1954CIS} have been
valuable tools in many areas, including reliability
\citep{OConnor2002Pre,Xie2002Sec}, medicine
\citep{Carey2003Ihw,Lawson2005Sas,Woodall2006Tuo} and finance
\citep{Frisen2008FS}.  See \cite{Stoumbos2000SoS} and the special
issues of ``Sequential Analysis'' (2007, Volume 26, Issues 2,3) for an
overview.  Often, heterogeneity between observations is accounted for
by using risk-adjusted charts based on fitted regression models
\citep{Grigg2004oor,Horvath2004Mci,Gandy2010ram}.

A common convention in monitoring based on control charts is to assume
the probability distribution of in-control data to be known.  In
practice this usually means that the distribution is estimated based
on a sample of in-control data and the estimation error is ignored.
Examples of this are
\cite{Steiner2000Msp,Grigg2004oor,Bottle2008Iin,Biswas2008rCi,Fouladirad2008Otu,Sego2009Rmo,Gandy2010ram}.

However, the estimation error has a profound effect on the performance
of control charts. This has been mentioned at several places in the
literature, e.g.\ \cite{jones2004rld,Albers2004Esc,jensen2006epe,Stoumbos2000SoS,Champ2007PoM}.

To illustrate the effect of estimation, we consider a CUSUM chart
\citep{Page1954CIS} with normal observations and estimated in-control
mean. We observe a stream of independent random variables
$X_1,X_2,\ldots$ which in control have an $ N(\mu,1)$ distribution and
out of control have an $N(\mu+\Delta,1)$ distribution, where $\Delta>0$
is the shift in the mean.  The chart
switches from the in-control state to the out-of-control state at an
unknown time $\kappa$.  The unknown in-control mean $\mu$ is estimated
by the average $\hat\mu$ of $n$ past in-control observations
$X_{-n},\dots,X_{-1}$ (this is often called phase 1 of the monitoring;
the running of the chart is called phase 2). We consider the  CUSUM chart
$$S_t=\max(0, S_{t-1}+X_t-\hat \mu - \Delta/2), \quad S_0=0 $$
with hitting time $\tau=\inf\{t>0: S_t\geq c\}$ for some threshold
$c>0$.

\begin{figure}[tb]
  \centering
  \includegraphics[width=0.95\linewidth]{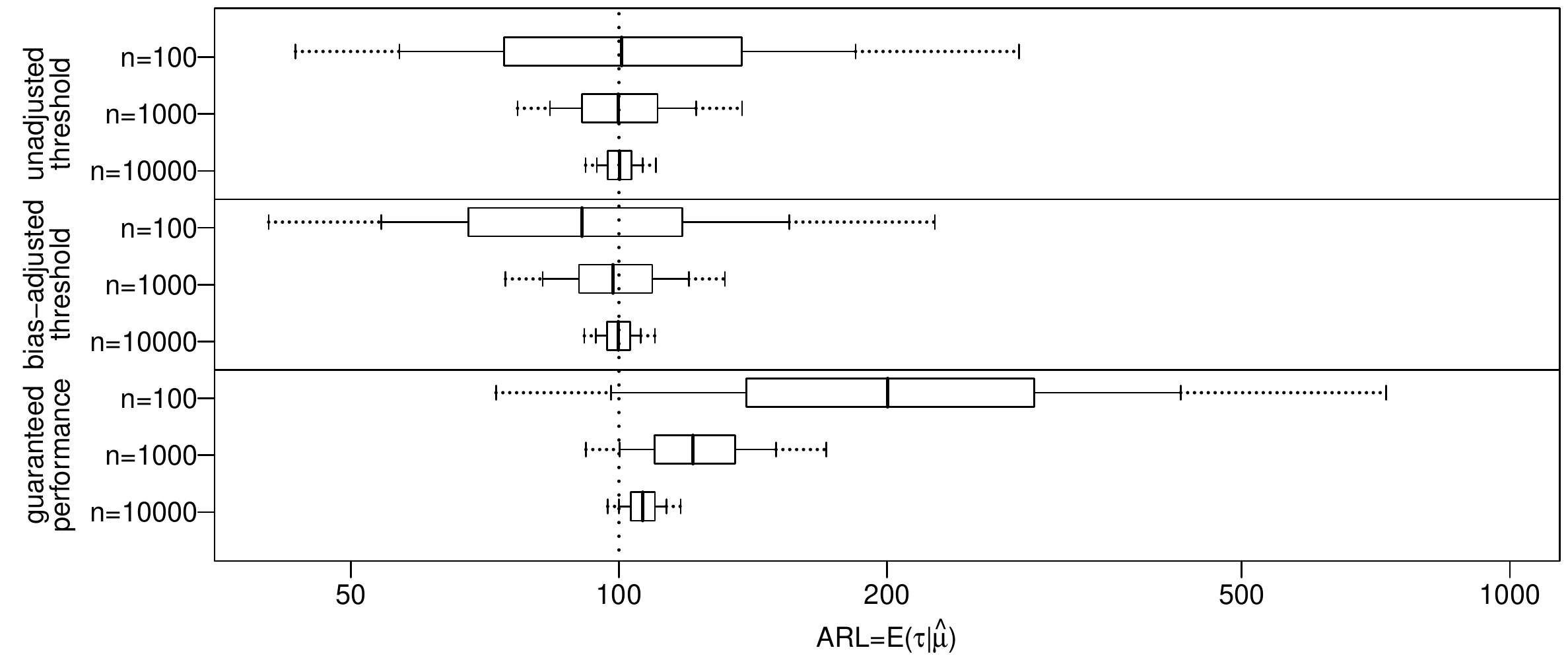}
  \caption{In-control distribution of ARL=$\E(\tau|\hat\mu)$ for
    CUSUMs for standard normally distributed data.  The mean $\hat
    \mu$ used in the monitoring is estimated based on $n$ past
    observations.  The boxplots show the 2.5\%, 10\%, 25\%, 50\%,
    75\%, 90\% and 97.5\% quantiles.The top part of the plot shows
    the situation when estimation error is ignored. In the middle part  the
    threshold has been chosen to give an unconditional ARL of 100
    (averaging out the parameter estimation). In the bottom part the threshold is
    adjusted to guarantee  with 90\% probability an in-control ARL of at least 100.  }
  \label{fig:estimerr}
\end{figure}

The in-control average run length,
$\ARL=\E(\tau|\hat\mu,\kappa=\infty)$, depends on $\hat \mu$ and is
thus a random quantity. The top part of the plot in
Figure~\ref{fig:estimerr} shows boxplots of its distributions with
threshold $c=2.84$, $\Delta=1$ and various numbers of past
observations. If  $\hat \mu=\mu$, i.e.\ $\mu$ was know,  this
would give an in-control $\ARL$ of 100. The estimation error is having a substantial effect on
the attained $\ARL$ even for  large samples such as
$n=1000$.
For further illustrations of the impact of estimation error see
\cite{jones2004rld} for CUSUM charts and  \cite{Albers2004Esc} for
Shewhart charts.

So far, no general approach for taking the estimation error into
account has been developed, but there are many special constructions
for specific situations.  For instance, for some charts so called
self-starting charts
\citep{Hawkins1987SCC,Hawkins1998csc,Sullivan2002SCC}, maximum
likelihood surveillance statistics
 to eliminate parameters
\cite[e.g.][]{Frisen2009}, correction
factors for thresholds
\citep{Albers2004Esc,jones2002statistical},
modified thresholds \citep{Zhang2011TSX} and threshold
functions \citep{Horvath2004Mci,Aue2006Cpm} have been
developed. Various bootstrap schemes for specific situations have also
been suggested, see for instance
\cite{Kirch2008BSC,Chatterjee2009Dcs,Capizzi2009Bdo,Huskova10Bsc}.
Further, some nonparametric charts which account for the estimation
error in past data have been proposed, see \cite{Chakraborti2007Ncc}
and references therein.   Recently some modified charts for monitoring
variance in the normal distribution with estimated parameters have
been suggested by \cite{Maravelakis2009AEC}  and
\cite{Castagliola2011ACC}.

When addressing estimation error, the above methods mainly  focus
on the performance of the charts averaged over both the estimation of
the in-control state as well as  running  the chart once.
In the middle part of Figure~\ref{fig:estimerr}, the threshold has
been chosen such that,  averaged over both the estimation of
the in-control state as well as  running  the chart once, the average
run length is $100$ (this results in a different threshold
for each $n$). It turns out that only a small change in the threshold
is needed and that the distribution of the conditional $\ARL=\E(\tau|\hat
\mu)$ is only changed slightly. This bias correction
for the $\ARL$ actually goes in the wrong direction in the sense that
it implies more short $\ARL$s. This is due to the $\ARL$ being
substantially influenced by the right tail of the run length
distribution, see the discussion in Section 2 of \cite{Albers2006SAC}.

However, usually, after the chart parameters are estimated, the
chart is run for some time without any reestimation of the in-control
state even if the chart signals. Moreover, in some situations, several
charts are run based on the same estimated parameters.
In these situations the ARL conditional on the estimated in-control
state is more relevant than the unconditional ARL. In the middle and
upper part of Figure \ref{fig:estimerr}, one sees that the conditional
ARL can be much lower than 100, meaning that both the unadjusted
threshold and the
 threshold adjusted for bias in the unconditional ARL
lead, with a substantial probability, to  charts
that have a considerably decreased time until false alarms.


To overcome these problems we will look at the performance of the
chart conditional on the estimated in-control distribution, averaging
only over different runs of the chart. This will lead to the
construction of charts that with high probability have an in-control
distribution with desired properties conditional on the observed past
data, thus reducing the situations in which there are many false
alarms due to estimation error.

The bottom part of Figure \ref{fig:estimerr}  shows the distribution of the in control
ARL when   the threshold for each set of past data is adjusted to
guarantee an in-control ARL of at least 100 with probability 90\%. The adjustment is
calculated using a bootstrap procedure explained later in the paper.
The adjustment succeeds to avoid the too low ARLs with the
prescribed probability, and we will see later that the cost in a
higher out-of-control ARL is modest.  Using hitting probabilities instead of ARL as
criterion leads to similar results.

\xx{talk about adjustments}

Our approach is similar in spirit to the exceedance probability
concept developed by Albers and Kallenberg for various types of
Shewhart \citep{Albers2004AEC,Albers2005Ncf,Albers2005EPF} and
negative binomial charts \citep{Albers2009CUM,Albers2010Toc}. They
calculate approximate adjusted thresholds such that there is only a
small prescribed probability that some performance measure, for
instance an ARL, will be a certain amount below or above a specified
target.

The main difference between their approach and what we present is that
our approach applies far more widely, to many different types of
charts and without having to derive specific approximation formulas in
each setting. If we apply a nonparametric bootstrap, the
proposed procedure will be robust against model misspecification.  In
addition to that, our approach allows not only to adjust the threshold
but also to give a confidence interval for the in-control performance
of a chart for a fixed threshold.  Lastly, even though not strongly
advocated in this paper, the bootstrap procedure we propose can also
be used to do a bias correction for the unconditional performance of
the chart, as in the middle part of Figure \ref{fig:estimerr}.


Next, we describe our approach more formally.
Suppose we want to use a monitoring scheme and that the in-control
distribution $P$ of the observations is unknown, but that
based on past in-control behaviour we have an estimate $\hat P$ of the
in-control distribution.  Let $q$ denote the in-control property of
the chart we want to compute, such as the $\ARL$, the false alarm
probability or the threshold needed for a certain $\ARL$ or false alarm
probability. In the above example we were interested to find a
threshold such that the in-control ARL is 100.

Generally, $q$ may depend on both the true in-control distribution $P$
and on estimated parameters of this distribution which for many charts
are needed to run the chart. We denote these parameters by $\hat
\xi=\xi(\hat P)$.
In the above CUSUM chart example $\hat \xi = \hat
\mu$.  We are interested in $q(P;\hat\xi)$, that
is the in-control performance of the chart conditional on the
estimated parameter. In the above CUSUM example,  $q(P;\hat\xi)$ is the threshold
needed to give an $\ARL$ of 100 if the observations are from the true
in-control distribution  $P$ and the
estimated parameter $\hat \mu$ is used. As $P$ is not observed $q(P;\hat
\xi)$ is not observable.  As mentioned above, many papers pretend
that the estimated in-control distribution $\hat P$ equals the true
in-control distribution $P$ and thus use $ q(\hat P;\hat \xi)$.
Our suggestion is to use bootstrapping of past data to construct an
approximate one-sided confidence intervals for $q(P;\hat \xi)$. From
this we get a  guaranteed conditional performance of the control
scheme.



In Section~\ref{sec:monitorhomobs} we present the general idea in the
setting with homogeneous observations, and discuss this for Shewhart
and CUSUM charts.  The main theoretical results are presented in
Section~\ref{sec:gentheor}, with most of the proofs given in the
Appendix.  Section \ref{sec:simulsingle} contains simulations
illustrating the performance of charts for homogeneous observations.
In Section~\ref{sec:regmod} extensions to charts based on regression
and survival analysis models are presented. Some concluding comments
are given in Section~\ref{sec:conclusion}.  The suggested methods are
implemented in a flexible R-package, that will be made available on the
Comprehensive R Archive Network (CRAN).

\section{Monitoring homogeneous observations}
\label{sec:monitorhomobs}
\subsection{General idea}
\label{subsec:generalidea}

Suppose that in control we have independent observations
$X_1,X_2,\dots$ following an unknown distribution $P$.  We want to use
some monitoring scheme/control chart that detects when $X_{i}$ is no
longer coming from $P$. The particular examples we discuss in
this paper are Shewhart and CUSUM charts, but the methodology we
suggest applies more widely.

To run the charts, one often needs certain parameters $\xi$. For
example, in the CUSUM control chart of the introduction, we
needed $
\xi= \mu$, the assumed in-control mean.  These parameters will usually
be estimated.

Let $\tau$ denote the time at which the chart signals a change. As $\tau$ may
depend on $\xi$, we sometimes write $\tau(\xi)$.  The
charts we consider use a threshold $c$, which determines how quickly
the chart signals (larger $c$ lead to a later signal).

The performance of such a control chart with the in-control
distribution $P$ and the parameters $ \xi$ can, for example, be
expressed as one of the following.
\begin{itemize}%
\setlength{\itemsep}{0pt}%
\setlength{\parskip}{0pt}%
\item $\ARL(P;\xi)=\E(\tau( \xi))$, where $\E$ is the expectation with respect to $P$.
\item ${\hit}(P;\xi)=\Prob(\tau( \xi)\leq T)$ for
  some finite $T>0$, where $\Prob$ is the probability measure under which $X_1,X_2,\dots\sim P$. This is the false alarm probability in $T$ time
  units.\xx{do we want to distinguish between $\Prob$ and $P$?}
\item $c_{\ARL}(P;\xi)=\inf\{c>0:\ARL(P;\xi)\geq\gamma\}$ for some
  $\gamma>0$. Assuming appropriate continuity, this is the threshold
  needed to give an in-control average run length of $\gamma$.
\item $c_{\hit}(P;\xi)=\inf\{c>0:\hit(P;\xi)\leq\beta\}$ for
  some $0<\beta<1$. This is the threshold needed
  to give a false alarm probability of $\beta$.
\end{itemize}
The latter two quantities are very important in practice, as they are
needed to decide which threshold to use to run a chart.  In the
notation we have suppressed the dependence of the quantities on $c$,
$T$, $\gamma$, $\beta$ and $\Delta$.

In the following, $q$ will denote one of $\ARL$, $\hit$, $c_{\ARL}$
or $c_{\hit}$, or simple transformations such as $\log(\ARL)$,
$\logit(\hit)$, $\log(c_{\ARL})$ and $\log(c_{\hit})$, where
$\logit(x)=\log\left(\frac{x}{1-x}\right)$.

The true in-control distribution $P$ and the parameters $\xi=\xi(P)$
needed to run the chart are usually estimated.  We assume that we have
past in-control observations $X_{-n},\dots,X_{-1}$ (independent of $X_1,X_2,\dots$), which
we use to  estimate the in-control distribution $P$ parametrically or
non-parametrically. We denote this estimate by $\hat P$.  The estimate
of $\xi$ will be denoted by $\hat \xi=\xi(\hat P)$. For example, in
the CUSUM control chart of the introduction, $\hat \xi=\hat \mu$ is
the estimated in-control mean.

The observed performance of the chart will depend on the true
in-control distribution $P$ as well as on the estimated parameters
$\hat \xi$ that are used to run the chart.  Thus we are interested in
$q(P;\hat\xi)$, the performance of the control chart
\emph{conditional} on $\hat \xi$.  This is an unknown quantity as $P$
is not known.  Based on the estimator $q(\hat P;\hat\xi)$, we
construct a one-sided confidence interval for this quantity to
guarantee, with high probability, a certain performance for the
chart. We choose to call the interval a  confidence interval,
even though the quantity $q(P;\hat\xi)$ is random.

We suggest the following for guaranteeing an upper bound on $q$ (which
is relevant for  $q=\hit$, $q=c_{\ARL}$ or $q=c_{\hit}$).  For $\alpha\in (0,1)$,
let $p_\alpha$ be a constant such that
$$
\Prob(q(\hat P;\hat\xi) - q(P;\hat\xi) >p_\alpha)=1-\alpha,
$$
assuming that such a $p_{\alpha}$ exists.
Hence,
$$
\Prob(q(P;\hat\xi)< q(\hat P;\hat\xi)-p_{\alpha})= 1-\alpha.
$$
Thus $(-\infty,q(\hat P;\hat\xi)-p_{\alpha})$ could be considered an
exact lower one-sided confidence interval of $q(P;\hat\xi)$.
\xx{Or,
 tolerance interval?}
\cc{This is not really a confidence interval in the classical sense -
 $q( P;\hat\xi)$ is an
  unobserved random quantity... and not just a fixed parameter.
The article  Weerahandi (1993, JASA) ``Generalized Confidence Intervals'' might be a useful reference. }

Of course, $p_{\alpha}$ is unknown.  We suggest to obtain an
approximation of $p_{\alpha}$ via bootstrapping.  In the following,
$\hat P^\ast$ denotes a parametric or non-parametric bootstrap
replicate of the estimated in-control distribution $\hat P$.
We can approximate $p_\alpha$ by $p^\ast_\alpha$ such that
$$\Prob(q(\hat P^\ast;\hat\xi^\ast)-q(\hat P;\hat\xi^\ast)> p^\ast_\alpha|\hat P)=1-\alpha.$$
\cc{alternatively we could use $\hat P(q(\hat P^\ast;\hat\xi^\ast)-q(\hat P;\hat\xi^\ast)> p^\ast_\alpha)=1-\alpha.$
}
Thus
\begin{equation}
  \label{eq:onesidedapproxconfint}
 (-\infty,q(\hat P;\hat\xi)-p^\ast_\alpha)
\end{equation}
is a one-sided (approximate) confidence interval for $q(P;\hat\xi)$.
In this paper, we will use the following generic algorithm to
implement the bootstrap.
\begin{algorithm}[Bootstrap]
  \label{alg:Bootstrap}
\hspace*{2mm}\\[-7mm]
\begin{enumerate}%
\setlength{\itemsep}{0pt}%
\setlength{\parskip}{0pt}%
\item From the past data $X_{-n},\dots,X_{-1}$, estimate
  $\hat P$ and  $\hat\xi$.
\item Generate bootstrap samples $X^{\ast}_{-n},\dots,X^{\ast}_{-1}$
  from $\hat P$.  Compute the corresponding estimate $\hat P^{\ast}$
  and $\hat \xi^{\ast}$. Repeat $B$ times to get $\hat
  P^{\ast}_1,\dots,\hat P^{\ast}_B$ and $\hat \xi^{\ast}_1,\dots,\hat
  \xi^{\ast}_B$.
\item Let $p_{\alpha}^{\ast}$ be the $1-\alpha$ empirical quantile of
  $q(\hat P^{\ast}_b;\hat\xi^{\ast}_b)-q(\hat P;\hat\xi^{\ast}_b)$, $b=1,\dots,B$.
\end{enumerate}
\end{algorithm}

For guaranteeing a lower bound on $q$, which is for example relevant
for $q=\ARL$, a similar upper one-sided confidence interval can be
constructed.

In a practical situation, the focus would be on deciding which
threshold to use for the control chart to obtain desired in-control
properties. We suggest to use either $q=c_{\ARL}$ or $q=c_{\hit}$, or
log transforms of these, and
then run the chart with the adjusted threshold
\begin{equation}
  \label{eq:adjThreshold}
q(\hat P;\hat \xi)-p^{\ast}_{\alpha}.
\end{equation}
This will guarantee that in (approximately) $1-\alpha$ of the
applications of this method, the control chart actually has the
desired in-control properties.

\cc{The following are some comments which are probably not quite relevant to practice.
In some application of control charts the chart parameters $\xi$ are
not estimated but determined according to certain specifications the
process should meet. A typical example would be industrial
applications like monitoring of properties of mass produced units
where there are precise specification of physical properties of the
units which should be monitored. Then $\xi$ may be determined
according to these specifications, and the point of the monitoring is
to detect deviations from the specifications.  However, the full
in-control distribution $P$ would usually still be unknown and our
approach would still apply for constructing confidence intervals for
$q(P;\xi_s)$ where $\xi_s$ denotes a specified $\xi$.

Would this in practice be relevant? Or would one also specify $P$,
  or at least parts of $P$ like the  mean?  Could using our approach
  here e.g.\ lead to picking a far too   large $c$ to get a guaranteed
  ARL if $P$ actually is far off from   where it ``should be''?
}

\subsection{Specific charts}

\subsubsection{Shewhart charts}
\label{subsubsec:Shewhart}

The one-sided Shewhart chart \citep{Shewhart1931ECo} signals at
$$
\tau=\inf\{t \in \{1,2,\dots\}: f(X_t,\xi)>c\}
$$
for some threshold $c$, where  $f$ is some function, $X_t$ is the
observation at time $t$ and $\xi$ are
some parameters.
 $X_t$ can be a single
measurement or e.g.\ the average, range or standard deviation of a
specified number of measurements, or some other statistic like a
proportion.
It is common to use a Shewhart chart with a threshold of
the mean plus 3 times the standard deviation, in
this case one would use $c=3$ and   $f(x,\xi)=\frac{x-\xi_1}{\xi_2}$
with $\xi_1$ being the mean and $\xi_2$ being the standard deviation.
 For two-sided charts one could just use  $f(x,\xi)=\frac{|x-\xi_1|}{\xi_2}$.

Conditionally on fixed parameters  $ \xi$, the stopping time $\tau$ follows a
geometric distribution with parameter
$p=p(c;P,\xi)=\Prob(f(X_t,\xi)>c)$.
Then the performance measures mentioned in the previous section simplify to
\begin{align*}
\ARL(P;\xi)=&\frac{1}{p(c;P,\xi)},& \hit(P;\xi)=&1-(1-p(c;P,\xi))^T,\\
 c_{\ARL}(P;\xi)=&p^{-1}\left( \frac{1}{\gamma}
   ;P,\xi\right)
\;\;\;\;\;\;
\text{ and}&
 c_{\hit}(P;\xi)=&p^{-1}\left( 1-(1-
    \beta)^{\frac{1}{T}};P,\xi\right),
\end{align*}
 where $p^{-1}(\cdot;P,\xi)$ is the inverse of $p(\cdot;P,\xi)$.
\cc{To get the formula for $c_{\hit}$ set
    $\beta=\hit$ in the second item and solve for $c$.}

Suppose that the in-control distribution comes from a parametric
family $P_{\theta}, \theta\in \Theta$.  Furthermore, suppose that we
have some way of computing an estimate $\hat\theta$ of $\theta$ based on the
sample.
Then we can use  Algorithm \ref{alg:Bootstrap} with $\hat P=P_{\hat \theta}$ to compute a
confidence interval as given by (\ref{eq:onesidedapproxconfint}).

Shewhart charts depend heavily on the tail behaviour of the
distribution of the observations. This is particularly problematic
when the sample size is small and we use non-parametric methods or a
simple non-parametric bootstrap. We thus primarily suggest to use a
parametric bootstrap for Shewhart charts.

\begin{remark}
In certain cases the parametric bootstrap will actually be exact when
$B \to \infty$. This happens when the distribution of
$q(P_{\hat\theta};\hat\xi) - q(P_{\theta};\hat\xi)$ under $P_{\theta}$
does not depend on $\theta$. In particular, this implies that
$q(P_{\hat\theta^\ast};\hat\xi^\ast) - q(P_{\hat\theta};\hat\xi^\ast)$
has the same distribution and $p^\ast_\alpha\to p_\alpha$ as $B\to
\infty$.

 As an example, consider the case when
 $f(x,\xi)=\frac{x-\xi_1}{\xi_2}$ and $X_t$ follows an
 $N(\xi_1,\xi_2^2)$ distribution and $q$ is any of the performance
 measures described above. We use $\theta=\xi$ and as estimator $\hat
 \xi_1$ we use the sample mean and as estimator $\hat\xi_2$ we use the
 sample standard deviation. Then
\begin{align*}
p(c;P_{\!\xi},\hat\xi)
=\Prob_{\!\xi}\!\left(\frac{X_t-\hat\xi_1}{\hat\xi_2}> c\right)
=1-\Phi\left(\frac{c\hat\xi_2+\hat\xi_1-\xi_1}{\xi_2}\right),
\end{align*}
where $\Phi$ is the cdf of the standard normal distribution, and
under $P_\xi$,
$$
\frac{c\hat\xi_2+\hat\xi_1-\xi_1}{\xi_2}=c\frac{\hat\xi_2}{\xi_2}+\frac{\hat\xi_1-\xi_1}{\xi_2}\sim\frac{c}{\sqrt{n-1}}\sqrt{W}+\frac{1}{\sqrt{n}}Z,
$$
where $W\sim\chi_{n-1}^2$ and $Z\sim N(0,1)$ are independent. Thus the
distribution of $p(c;P_\xi,\hat\xi)$, and hence $q(P_{\xi};\hat\xi)$,
is completely known.  As $
p(c;P_{\!\hat\xi},\hat\xi)=\Prob_{\!\hat\xi}\left(\frac{X_t-\hat\xi_1}{\hat\xi_2}>
  c\right) =1-\Phi(c)$, and thus $q(P_{\!\hat\xi};\hat\xi)$, is not
random, the distribution of $q(P_{\!\hat\xi};\hat\xi) -
q(P_{\xi};\hat\xi)$ also does not depend on any unknown
parameters. Thus the parametric bootstrap is exact in this example.
\end{remark}

\subsubsection{CUSUM charts}
\label{subsubsec:CUSUM}
This section considers the one-sided CUSUM chart \citep{Page1954CIS}.
The classical CUSUM chart was designed to detect a shift of size
$\Delta>0$ in the mean of normally distributed observations. Let $\mu$ and $\sigma$
denote, respectively, the in-control mean and standard
deviation. A CUSUM chart can be defined by
\begin{equation}
  \label{eq:discrCUSUM_meanshift}
S_t=\max(0, S_{t-1}+(X_t-\mu - \Delta/2)/\sigma), \quad S_0=0
\end{equation}
with hitting time $\tau=\inf\{t>0: S_t\geq c\}$ for some threshold
$c>0$.

Alternatively, we could drop the scaling and not divide by the
 standard deviation $\sigma$ in
(\ref{eq:discrCUSUM_meanshift}). See Chapter 1.4 in
\cite{Hawkins1998csc} for a discussion on scaled versus unscaled
CUSUMs.

More generally, to accommodate observations with general in-control distribution with
density $f_0$ and general out-of-control distribution with density $f_{1}$, it
is optimal in a certain sense \citep{Moustakides1986OST} to modify the
CUSUM chart by replacing $(X_t-\mu - \Delta/2)/\sigma$ by the log
likelihood ratio $\log(f_1(X_t,\theta)/f_0(X_t, \theta))$ such
that the CUSUM chart is
\begin{equation}
  \label{eq:discrCUSUM_loglikelihood}
S_t=\max(0, S_{t-1}+\log(f_1(X_t,\theta)/f_0(X_t, \theta))), \quad S_0=0.
\end{equation}

Let $\xi$ denote either $(\mu,\sigma)$ in
(\ref{eq:discrCUSUM_meanshift}) or $\theta$ in
(\ref{eq:discrCUSUM_loglikelihood}).  Usually, $\xi$ needs to be
estimated from past data, and we can then use Algorithm
\ref{alg:Bootstrap} to compute a confidence interval
(\ref{eq:onesidedapproxconfint}) for the performance measure
$q(P;\hat\xi)$. For (\ref{eq:discrCUSUM_loglikelihood}) it is most
natural to use a parametric bootstrap with $\hat P=P_{\hat \theta}$,
while for (\ref{eq:discrCUSUM_meanshift}) we can use either a
parametric or a nonparametric bootstrap. In the latter case we let
$\hat P$ be the empirical distribution of $X_{-n},\dots,X_{-1}$, i.e.
in Algorithm \ref{alg:Bootstrap}, $X^{\ast}_{-n},\dots,X^{\ast}_{-1}$
are sampled with replacement from $X_{-n},\dots,X_{-1}$.

\begin{remark}
Similar as for Shewhart charts, this parametric bootstrap is exact
when the distribution of
$q(P_{\hat\theta};\hat\xi)-q(P_{\theta};\hat\xi)$ does not have any
unknown parameters.  This is, for instance, the case if we use
(\ref{eq:discrCUSUM_loglikelihood}) for an exponential distribution
with the out-of-control distribution specified as an exponential
distribution with mean $\Delta\lambda$, where $\lambda$ is the
in-control mean. Another example of this is when we have normally
distributed data and use a CUSUM with the increments
$(X_t-\hat\mu)/\hat\sigma-\Delta/2$.
\end{remark}

\section{General theory}
\label{sec:gentheor}

In this section, we show that asymptotically, as the number of past
observations $n$ increases, our procedure works.  An established way
of showing asymptotic properties of bootstrap procedures is via a
functional delta method \citep{Vaart1996WCa,Kosorok2008ItE}. Whilst we
will follow a similar route, our problem does not fit directly into
the standard framework, because the quantity of interest, $q(P,\hat
\xi)$, contains the random variable $\hat \xi$.
We present the setup and
the main result in Section \ref{sec:th:main}, followed by examples
(Section \ref{sec:th:examples}).
\cc{The asymptotic development in this section only show that things
  do not go badly wrong as $n\to \infty$. They only establish
  consistency of the correction/confidence intervals.  However, the need to use
  these confidence intervals disappears as $n$ increases.}

\subsection{Main theorem}
\label{sec:th:main}
Let $D_q$ be the set in which  $P$ and its
estimator $\hat P$ lie, i.e.\ a set describing the potential
probability distribution of our observations. This could be a subset
of $\R^d$ for parametric distributions, the set of cumulative
distribution functions for non-parametric situations, or the set of
joint distributions of covariates and observations.  We assume that
$D_q$ is a subset of a complete normed vector space $D$. \cc{Do we
  want to /need to assume that $D$ is complete (every Cauchy sequence
  converges)?  This should not be a problem as $R^k$ and
  $l_{\infty}(\R)$ are complete metric spaces. } Let $\Xi$ be a
non-empty topological space containing the potential parameters $\xi$
used for running the chart. In our examples, we will let $\Xi\subset
\R^d$ be an open set.

We assume that $\hat P^{\ast}=\hat P^{\ast}(\hat P, W_n)$ is a
bootstrapped version of $\hat P$ based both on the observed data $\hat P$ and
on an independent random vector $W_n$.  For example, when resampling
with replacement then $W_n$ is a weight vector of length $n$,
multinomially distributed, that determines how often a given
observation is resampled.  In a parametric bootstrap, $W_n$ is the
vector of random variables needed to generate observations from the
estimated parametric distribution.

In the main theorem we will need that the mapping $q:D_q\times
\Xi\to\R$, which returns the property of the chart we are interested
in, satisfies the following extension of Hadamard differentiability.
For the usual definition of Hadamard differentiability see e.g.\
\citep[Section 20.2]{Vaart1998AS}. The extension essentially consists in
requiring Hadamard differentiability in the first component when the second
component is converging.
\begin{definition}
\label{def:haddiffamily}
Let $D,E$ be metric spaces, let $D_f\subset D$ and let $\Xi$ be a
non-empty topological space. \cc{We need at least to be able to speak
  about convergence in $\Xi$.}  The family of functions
$\{f(\cdot;\xi):D_f \to E: \xi\in \Xi\}$ is called \emph{Hadamard
  differentiable at $\theta\in D_f$ around $\xi \in \Xi$ tangentially
  to $D_0\subset D$}  if there exists a continuous linear map \cc{this is a requirement that also appears in the original definition and which we may be using in our proofs; however, we never prove for our derivatives that they are continuous and linear}
$f'(\theta;\xi):D_0\to E$ such that
$$
\frac{f(\theta+t_nh_n;\xi_n)-f(\theta;\xi_n)}{t_n}\to
f'(\theta;\xi)(h)\quad(n\to\infty)
$$
for all sequences $(\xi_n)\subset \Xi$, $(t_n)\subset \R$, $(h_n)\subset D$
that satisfy $\theta+t_nh_n\in D_f \,\forall n$ and $\xi_n\to \xi$, $t_n\to 0$, $h_n\to h\in D_0$ as $n\to \infty$.
\end{definition}

In the following theorem we understand convergence in distribution, denoted by $\leadsto$, as defined
in \citet[Def 1.3.3]{Vaart1996WCa} or in \citet[p.108]{Kosorok2008ItE}.
\cc{ Let
  $(\Omega_n, {\cal A}_n, P_n)$ be a sequence of probability spaces,
  let $(\Omega, \cal A, P)$ be a further probability space, let $D$ be
  a metric space and let $X_n:\Omega_n\to D$ be a sequence of maps and
  let $X:\Omega\to D$ be a Borel measurable map. Then $X_n\leadsto X$
  if $\E^{\ast}f(X_n) \to \E f(X)$ for all continuous, bounded $f:D\to
  \R$.  }
\cc{Outer expectation is defined in \cite{Vaart1996WCa} and
  in \cite{Kosorok2008ItE}, essentially $\E^{\ast}X = \inf \{\E Y:
  Y\geq X, Y \text{ measurable}\}$.}
\begin{theorem}
\label{th:main}
  Let $q:D_q\times \Xi\to \R$ be a mapping, let $P\in D_q$ and let
  $\xi:D_q\to \Xi$ be a continuous function.
  Suppose that the following conditions are satisfied.
\begin{itemize}%
\setlength{\itemsep}{0pt}%
\setlength{\parskip}{0pt}%
\item[a)] $q$ is Hadamard differentiable at $P$ around $\xi$ tangentially to $D_0$ for some $D_0\subset D$.
\item[b)] $\hat P$ is a sequence of random elements in $D_q$ such that
$
  \sqrt{n}(\hat P-P)\leadsto Z
$ as $n\to \infty$
where $Z$ is some tight random element in $D_0$.
\item[c)]
$\sqrt{n}(\hat P^{\ast}-\hat P)\condweakconv{\hat P} Z$ as $n\to\infty$
where $\condweakconv{\hat P}$ denotes weak convergence conditionally on $\hat
P$ in probability as defined in \citet[p.19]{Kosorok2008ItE}.  \cc{
  i.e. $\sup_{h\in \text{BL}_1}|E_Wh(\hat X_n) - E
  h(X)|\stackrel{P}{\to}0$ and $E_Wh(\hat X_n)^{\ast}-E_Wh(\hat
  X_n)_{\ast}\stackrel{P}{\to}0$ for all $f\in \text{BL}_1$ where the
  subscript $W$ denotes conditional expectation over the weights given
  the remaining data.  }
\item[d)] The cumulative distribution function  of  $q'(P;\xi)Z$ is continuous.
\item[e)] Outer-almost surely, the map $W_n\mapsto h(\hat P^{\ast}(\hat P, W_n))$ is measurable for each $n$ and for every continuous bounded function $h:D_q\to \R$.
\item[f)] $q(\hat P; \hat \xi)-q(P;\hat \xi)$ and  $p_\alpha^{\ast}$ are  random variables, i.e.\ measurable,
where $\hat \xi =\xi(\hat P)$ and $p^{\ast}_{\alpha}=\inf\{t\in \R:
\hat \Prob(q(\hat P^\ast;\hat\xi^\ast)-q(\hat P;\hat\xi^\ast)\leq t)\geq
\alpha\}$.
\end{itemize}
Then
\begin{equation*}
\Prob(q(P;\hat\xi)\in (-\infty,  q(\hat P;\hat\xi)-p^{\ast}_{\alpha}))\to  1-\alpha \quad (n\to \infty).
\end{equation*}
\end{theorem}
A similar result  holds for upper confidence intervals.

The proof is in Appendix \ref{sec:proof}.  The theorem essentially is
an extension of the delta-method.  Condition a) ensures the necessary
differentiability.  Conditions b) and c) are standard assumptions for
the functional delta method; b) for the ordinary delta method and c)
for the bootstrap version of it. Condition d) ensures that, after
using an extension of the delta-method, the resulting confidence
interval will have the correct asymptotic coverage probability.
Condition e) is a technical measurability condition, which will be
satisfied in our examples. Condition f) is a measurability condition,
which should usually be satisfied.

\subsection{Examples}
\label{sec:th:examples}
The following sections give examples in which Theorem \ref{th:main}
applies.  We consider hitting probabilities ($q=\hit$) and thresholds to
obtain certain hitting probabilities ($q=c_{\hit}$).

These examples are
meant to be illustrative rather than exhaustive. For example, other
parametric setups could be considered along similar lines to Section
\ref{sec:cusum-charts-with}. Furthermore, other performance measures such as
$\log(c_{\hit})$ or $\logit(\hit)$ would essentially require application of chain rules
to show differentiability.

\subsubsection{Simple nonparametric setup for CUSUM charts}
\label{sec:theor:ex:CUSUM:nonpar}
We show how the above theorem applies to the CUSUM chart described in
(\ref{eq:discrCUSUM_meanshift}) when using a non-parametric bootstrap
version of Algorithm \ref{alg:Bootstrap}.

Let $D=l_{\infty}(\R)$ be the set of bounded functions $\R\to\R$
equipped with the sup-norm $\|x\|=\sup_{t\in \R}|x_t|$. \cc{This is a
  Banach space, i.e. a complete normed vector space}
Let $D_q\subset D$ be the set of cumulative distribution functions on
$\R$ with finite second moment.
The parameters needed to run the chart are the mean and the standard deviation of the in-control observations, thus we may choose
 $\Xi=\R\times(0,\infty)$ and  $\xi:D_q\to \Xi, P\mapsto (\int x P(dx),
\int x^2 P(dx)-(\int x P(dx))^2)$.

As quantities $q$ of interest we are considering hitting probabilities
($q=\hit$) and thresholds ($q=c_{\hit}$) needed to achieve a certain hitting
probability.  The probability $\hit:D_q\times \Xi\to \R$ of hitting a
threshold $c>0$ up to step $T>0$ can be written as
$\hit(P;\xi)=\Prob(m(Y) \geq c)$, where
$m(Y)=\max_{i=1,\dots,T}R_i(Y)$ is the maximum value of the chart up
to time $T$,
$R_i(Y)=\sum_{j=1}^iY_j-\min_{0\leq k\leq i}\sum_{j=1}^kY_j$ is the
value of the CUSUM chart at time $i$, $Y=(Y_1,\dots,Y_T)$,
$Y_t=\frac{X_t-\xi_1-\Delta/2}{\xi_2}$ and $X_1,\dots,X_T \sim P$ are the
independent observations.  The threshold needed to achieve a certain hitting
probability $\beta \in (0,1)$ is $c_{\hit}: D_q\times \Xi\to \R$,
$c_{\hit}(P;\xi)=\inf\{c>0:\hit(P;\xi)\leq \beta\}$.

The setup for the nonparametric bootstrap is as follows. $W_{n}$ is an
$n$-variate multinomially distributed random vector with probabilities
$1/n$ and $n$ trials. The resampled distribution is $\hat
P^{\ast}=\frac{1}{n}\sum_{j=1}^nW_{nj}\delta_{X_{-j}}$, where $\delta_x$
denotes the Dirac measure at $x$.

The following lemma shows
condition a) of  Theorem \ref{th:main},
the Hadamard differentiability of $\hit$ and $c_{\hit}$.
\begin{lemma}
\label{le:HaddiffCUSUMhit}
For every $P\in D_q$, and every $\xi \in \R\times(0,\infty)$, the
function $\hit$ is Hadamard differentiable at $P$ around $\xi$
tangentially to $D_0=\{H:\R\to \R: H\text{ continuous}, \lim_{t\to
  \infty}H(t)=\lim_{t\to-\infty}H(t)=0\}$.  If, in addition, $P$ has a
continuous bounded positive derivative $f$ with $f(x)\to 0$ as $x\to
\pm \infty$, then $c_{\hit}$ is also Hadamard differentiable at $P$
around $\xi$ tangentially to $D_0$.
\end{lemma}
The proof  is in Appendix \ref{sec:haddifhitprobex}, with preparatory results in
Appendix \ref{sec:chain-rule} - \ref{sec:diff-hitt-prob}.

Conditions b) and c) of Theorem \ref{th:main} follow directly from empirical process theory,
see e.g.\  \cite[p.17,Theorems 2.6 and 2.7]{Kosorok2008ItE}.
\cc{To see conditions b) and c) of  Theorem \ref{th:main}, we can argue as follows.
  In the language of empirical process theory, consider ${\cal
    F}=\{\R\to \R, x\mapsto 1_{(-\infty,a]}(x):a\in \R\}$ and let
  $l_{\infty}({\cal F})$ be the set of all bounded function ${\cal
    F}\to\R$.  As $\cal F$ can be identified with $\R$, we can
  idenfity $l_{\infty}({\cal F})$ with $l_{\infty}(\R)$, the set of
  bounded functions $\R\to \R$.  By \cite[p.17]{Kosorok2008ItE}, $\cal F$ is
  Donsker, i.e.\ if $X_1,\dots,X_n\sim P$ independently, and letting
  $P_n=\frac{1}{n}\sum_{i=1}^n\delta_{X_i}$ be the corresponding
  empirical measure, then $G_n=\sqrt{n}(P_n-P)\leadsto G$ in
  $l_{\infty}(\cal F)$ (or equivalently in $l_{\infty}(\R)$) for some
  $G$.

  thus $G_n$ is considered a random element in $l_{\infty}({\cal
      F})$ (or equivalently $l_{\infty}(\R)$), via
    $\sqrt{n}(P_n-P)(1_{(-\infty,a]})=\sqrt{n}(P_n((-\infty,a])-P((-\infty,a]))$

  Now, Theorems 2.6 and 2.7 of \cite{Kosorok2008ItE} give conditional convergence
  results for the nonparametric bootstrap, i.e. they show that
$\hat G_n\condweakconv{\hat P} G$ in $l_{\infty}({\cal F})$ and that the sequence $\hat G_n$ is asymptotically measurable.
Sufficient conditions for c)  are e.g.\ given in Theorems 3.6.1 and 3.6.2 on p.347 of \cite{Vaart1996WCa}
}
Condition e)  is satisfied as well, see  bottom of p.189  and after Theorem 10.4 (p.184) of \cite{Kosorok2008ItE}.

Verifying condition d) in full is  outside the scope of the present paper.
A starting point could be the fact that by
the Donsker theorem, $Z\sim G\circ P$, where $G$ is a Brownian bridge.
\cc{We would need to consider the derivative in Lemma
  \ref{le:diffhitprob} and in Lemma \ref{le:Haddiffinversemap}.}

\subsubsection{CUSUM charts with normally distributed observations}
\label{sec:cusum-charts-with}
In this section, we consider a similar setup to the monitoring based
on (\ref{eq:discrCUSUM_meanshift}) considered in the previous
subsection with the difference that we now use parametric assumptions.
More specifically,  the observations $X_i$ follow a normal
distribution with unknown mean $\mu$ and variance $\sigma^2$. We will
use this both for computing the properties of the chart as well as in
the bootstrap, which will be a parametric bootstrap version of
Algorithm \ref{alg:Bootstrap}.

The distribution of the observations can be identified with its
parameters which we estimate by $\hat P = (\hat \mu, \hat \sigma^2)$,
where $\hat \mu=\frac{1}{n}\sum_{i=1}^nX_{-i}$ and $\hat
\sigma^2=\frac{1}{n-1}\sum_{i=1}^n(X_{-i}-\hat \mu)^2$.  The set of
potential parameters is $D_q=\R\times(0,\infty)$ which is a subset of
the Euclidean space $D=\R^2$. The parameters needed to run the chart
(\ref{eq:discrCUSUM_meanshift}) are just the same as the one needed to
update the distribution, thus $\Xi=D_q$ and $\xi:D_q\to \Xi,
(\mu,\sigma)\mapsto (\mu,\sigma)$ is just the identity.

As before, we are interested in hitting probabilities within the first
$T$ steps.  Using the function $\hit$ defined in the previous
subsection, we can write the hitting probability in this parametric
setup as $\hit^N:D_q\times \Xi\to\R$, $(\mu,\sigma;\xi)\mapsto\hit(
\Phi_{\mu,\sigma^2};\xi)$, where $\Phi_{\mu,\sigma^2}$ is the cdf of
the normal distribution with mean $\mu$ and variance $\sigma^{2}$ and
the superscript $N$ stands for normal distribution.  Furthermore,
using $c_{\hit}$ from the previous subsection, the threshold needed to
achieve a given hitting probability is $c_{\hit}^N:D_q\times\Xi\to\R$,
$(\mu,\sigma;\xi)\mapsto c_{\hit}(\Phi_{\mu,\sigma^2};\xi)$.

The resampling is a parametric resampling. To put this in the framework of the main theorem, we let $W_n=(W_{n1},\dots,W_{nn})$, where
$W_{n1},\dots,W_{nn}\sim N(0,1)$ are independent. The
resampled parameters are then  $\hat
\mu^{\ast}_n=\frac{1}{n}\sum_{i=1}^nX_{ni}^{\ast}$ and $\hat \sigma^{\ast
  2}_n=\frac{1}{n-1}\sum_{i=1}^n(X^{\ast}_{ni}-\hat \mu^{\ast}_n)^2$
where $ X^{\ast}_{ni}=\hat P_2W_{ni}+\hat P_1$.

The following lemma shows that condition a) of  Theorem \ref{th:main} is  satisfied.
\begin{lemma}
\label{le:HaddiffCUSUMhitNORMAL}
For every $\theta\in \R\times (0,\infty)$ and every  $\xi \in \R\times(0,\infty)$,
the functions $\hit^N$ and $c_{\hit}^{N}$ are Hadamard differentiable at $\theta$ around $\xi$.
\end{lemma}
The proof can be found in Appendix \ref{sec:haddifhitprobex}, using again the preparatory results of
Appendix \ref{sec:chain-rule} - \ref{sec:diff-hitt-prob}.

Concerning the other conditions of Theorem \ref{th:main}: Condition b)
can be shown using standard asymptotic theory, e.g.\ maximum likelihood
theory, which will yield that $Z$ is normally distributed. \cc{could
  argue via the $(\hat \mu, (n-1)/n\hat \sigma^2)$ being the MLE}
Condition c) is essentially the requirement that the parametric
bootstrap of normally distributed data is working.  \cc{This should be
  easy to shown by arguing conditionally on the estimators. There may
  be something in vdVaart, asymptotic statistics - but he is just
  using nonparametric resampling.}  As $Z$ is a normally distributed
vector, condition d) holds unless $q'$ equals 0.  Condition e) is
satisfied, as the mapping $W_n\mapsto\hat P^{\ast}(\hat P, W_n) =
(\hat\mu_n^{\ast}, \hat\sigma_n^{\ast 2})$ is continuous and hence
measurable.

\subsubsection{Setup for Shewhart charts}
For Shewhart charts, the same setup as in the previous two sections
can be used, the only difference is the choice of $q$. Conditions b),
c) and e) are as in the previous two sections.  We conjecture that it is possible to show the  Hadamard
differentiability  more directly, as the properties are
available in closed form, see Section \ref{subsubsec:Shewhart}.

\cc{
With $G=1-p$,
 $\hit(G)=(c\mapsto 1-G(c)^T)$,  $\hit'(G)(H)=(c\mapsto -G(c)^{T-1}H(c))$,

 $\frac{\partial}{\partial c} \hit(G)(c)=-T G(c)^{T-1}g(c)$

 $\ARL(G)=(c\mapsto\frac{1}{1-G(c)})$, $\ARL'(G)(H)=(c\mapsto\frac{1}{(1-G(c))^2}H(c))$ \cc{see \cite[Lemma 3.9.25]{Vaart1996WCa}}
 $\frac{\partial}{\partial c}\ARL(G)(c)=\frac{1}{(1-G(c))^2}g(c)$
}

\section{Simulations for homogeneous observations}
\label{sec:simulsingle}

We now illustrate our approach by some simulations using
CUSUM charts. The simulations were done in R \citep{R}.

We use two past sample sizes, $n=50$ and
$n=500$. The in-control distribution of $X_t$ is $N(0,1)$ and
we use 1000 replications and $B=1000$ bootstrap
replications. We employ both the parametric bootstrap and the
nonparametric bootstrap mentioned in the previous sections.  For the
parametric bootstrap we used the sample mean and sample standard
deviation of $X_{-n},\dots,X_{-1}$ as estimates for the mean and the standard deviation of
the observations.

For the performance measures $\ARL$, $\log(\ARL)$, $\hit$ and
$\logit(\hit)$ we use a threshold
of $c=3$. For $c_{\ARL}$ we calibrate to  an $\ARL$ of $100$
in control and for $c_{\hit}$ we calibrate to a false alarm probability of
$5\%$ in 100 steps.

We use the CUSUM chart (\ref{eq:discrCUSUM_meanshift}) with $\Delta=1$
and $\mu$ and $\sigma$ estimated from the past data. To compute
properties such as $\ARL$ or hitting probabilities, we  use a
Markov chain approximation (with 75 grid points), similar to the one
suggested in \cite{BROOK1972atp}\cc{there is a precise description of
  a grid that is being used in that paper - I think we are using
  something similar but most likely not completely identical}. \cc{We
checked that this gave very good  approximations.}

\subsection{Coverage probabilities}
\label{subsec:simulcoverage}

Table \ref{tab:covprob_simnormal} contains coverage probabilities of
nominal 90\% confidence intervals. These are the one-sided lower confidence
intervals given by (\ref{eq:onesidedapproxconfint}), except for
$q=\ARL$ and $\log(ARL)$ where the corresponding upper interval is used.

\begin{table}
    \caption{Coverage probabilities of nominal 90\% confidence intervals for CUSUM charts.
  \label{tab:covprob_simnormal}}
\begin{center}
\parbox{0.58\textwidth}{
    \begin{tabular}{|c| c| c| c| c| }\hline
&\multicolumn{2}{|c|}{Parametric }& \multicolumn{2}{c|}{Nonparametric} \\
\hline
q=& n=50&n=500&n=50&n=500 \\
\hline
$\ARL$&1.000&0.929&0.999&0.944\\
$\log(\ARL)$&0.928&0.899&0.902&0.915\\
$\hit$&0.923&0.896&0.878&0.910\\
$\logit(\hit)$&0.892&0.893&0.870&0.904\\
$c_{\ARL}$&0.881&0.892&0.846&0.893\\
$\log(c_{\ARL})$&0.896&0.895&0.868&0.904\\
$c_{\hit}$&0.878&0.890&0.843&0.891\\
$\log(c_{\hit})$&0.897&0.893&0.856&0.901\\
\hline\end{tabular}
\\
      The standard deviation of the results is roughly 0.01.
 }
    \end{center}
\end{table}

In the parametric case, for $n=50$, the coverage probabilities are
somewhat off for untransformed versions, in particular for $q=\ARL$.
Using $\log$ or $\logit$ transformations seems to improve the coverage
probabilities considerably.  In the parametric case, for $n=500$, all
coverage probabilities seem to be fine, except for $q=\ARL$, which
although  shows some marked improvement compare to $n=50$.  In the
nonparametric case, a similar picture emerges, but the coverage
probabilities are a bit worse than in the parametric case.


\begin{remark}
\label{rem:scalingdoesnotmatter}
  For $q= \log(c_{\ARL})$ and $q= \log(c_{\hit})$ the division by
  $\hat\sigma$ in
  (\ref{eq:discrCUSUM_meanshift}) could be skipped without making a
  difference to the coverage probabilities.  Indeed, the division by
  $\hat\sigma$ just scales the chart (and the resulting threshold) by
  a multiplicative factor, which is turned into an additive factor by
  $\log$ and which then cancels out in our adjustment.
\end{remark}

\subsection{The benefit of an adjusted threshold}
In this section, we consider both the in- and out-of-control
performance of CUSUM charts when adjusting the threshold $c$ to give a
guaranteed in-control $\ARL$ of 100.  Setting the threshold is, in our
opinion, the most important practical application of our method.

\begin{figure}[tb]
  \centering
  \includegraphics[width=\linewidth]{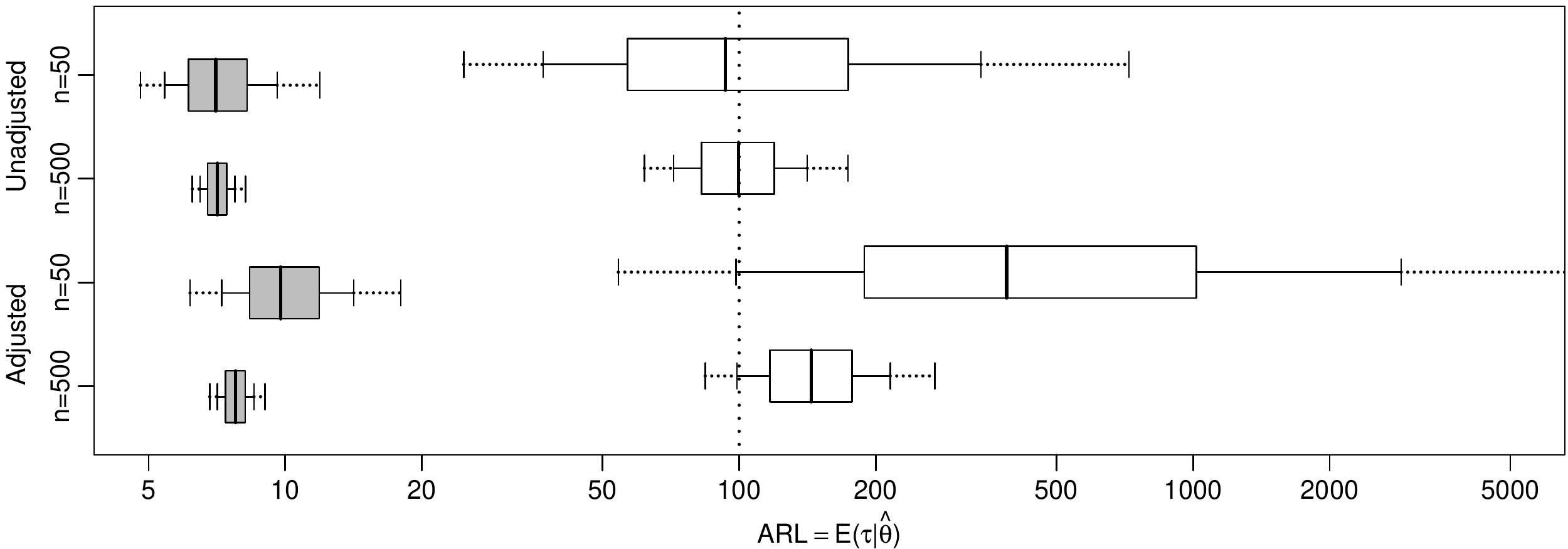}
  \caption{Distribution of the conditional $\ARL$ for CUSUMs in a
    normal distribution setup. Thresholds are calibrated to an
    in-control $\ARL$ of 100.  The adjusted thresholds have a
    guarantee of 90\%.  A log transform is used in the calibration.
    The boxplots show the 2.5\%, 10\%, 25\%, 50\%, 75\%, 90\% and
    97.5\% quantiles.  The white boxplots are in-control, the gray
    boxplots out-of-control.}
  \label{fig:adjusted_unadjusted_ARL}
\end{figure}

Figure \ref{fig:adjusted_unadjusted_ARL} shows average run lengths
for both  the unadjusted threshold $c(\hat P;\hat \mu, \hat \sigma)$ and the
adjusted threshold $\exp(\log( c(\hat P;\hat \mu, \hat \sigma))-p^{\ast}_{0.1})$, where $p^{\ast}_{0.1}$ is computed via the parametric
bootstrap using $q=\log(c_{\ARL})$.  Thus, with  90\%  probability, the  adjusted threshold should
lead to an $\ARL$ that is above 100. In this and in all  following
simulations,
the out-of-control ARL refers
to the  situation where the chart is  out-of-control from the
beginning, i.e.\ from
time 0 onwards.

For the unadjusted threshold, the desired in-control average run
length is only reached in roughly half the cases.  More importantly,
for $n=50$, the probability of having an in-control $\ARL$ of below $50$
is greater than 20\%.

With the adjusted threshold we should get an average run length of at
least 100 in 90\% of the cases.  This is achieved.  The
out-of-control $\ARL$ using the adjusted thresholds increases only
slightly compared to the unadjusted version.

Similarly to Remark \ref{rem:scalingdoesnotmatter}, removing the
scaling by $\hat \sigma$ in (\ref{eq:discrCUSUM_meanshift}) would not
change the results of this section.

\subsection{Nonparametric bootstrap - advantages and disadvantages}

In this section, we compare the parametric and the non-parametric
bootstrap.  We consider CUSUM charts that are calibrated to an
in-control average run length of 100 assuming a normal distribution.
We use the adjusted threshold $\exp(\log( c_{\ARL}(\hat P;\hat
\mu, \hat \sigma))-p^{\ast}_{0.1})$.

Figure \ref{fig:par_nonpar_ARL} shows the distribution of $\ARL$ for
$n=50$ and $n=500$ for both the parametric bootstrap that assumes a normal
distribution of the updates and the nonparametric bootstrap. We consider both a
correctly specified model where $X_t\sim N(0,1)$ as well as two
misspecified models where $X_t\sim \text{Exponential}(1)$ and $\sqrt{20}X_t\sim
\chi^2_{10}$ (all of the $X_t$ have variance 1).  We show both the in- as well as the
out-of-control performance of the charts.

\begin{figure}[tb]
  \centering
  \includegraphics[width=\linewidth]{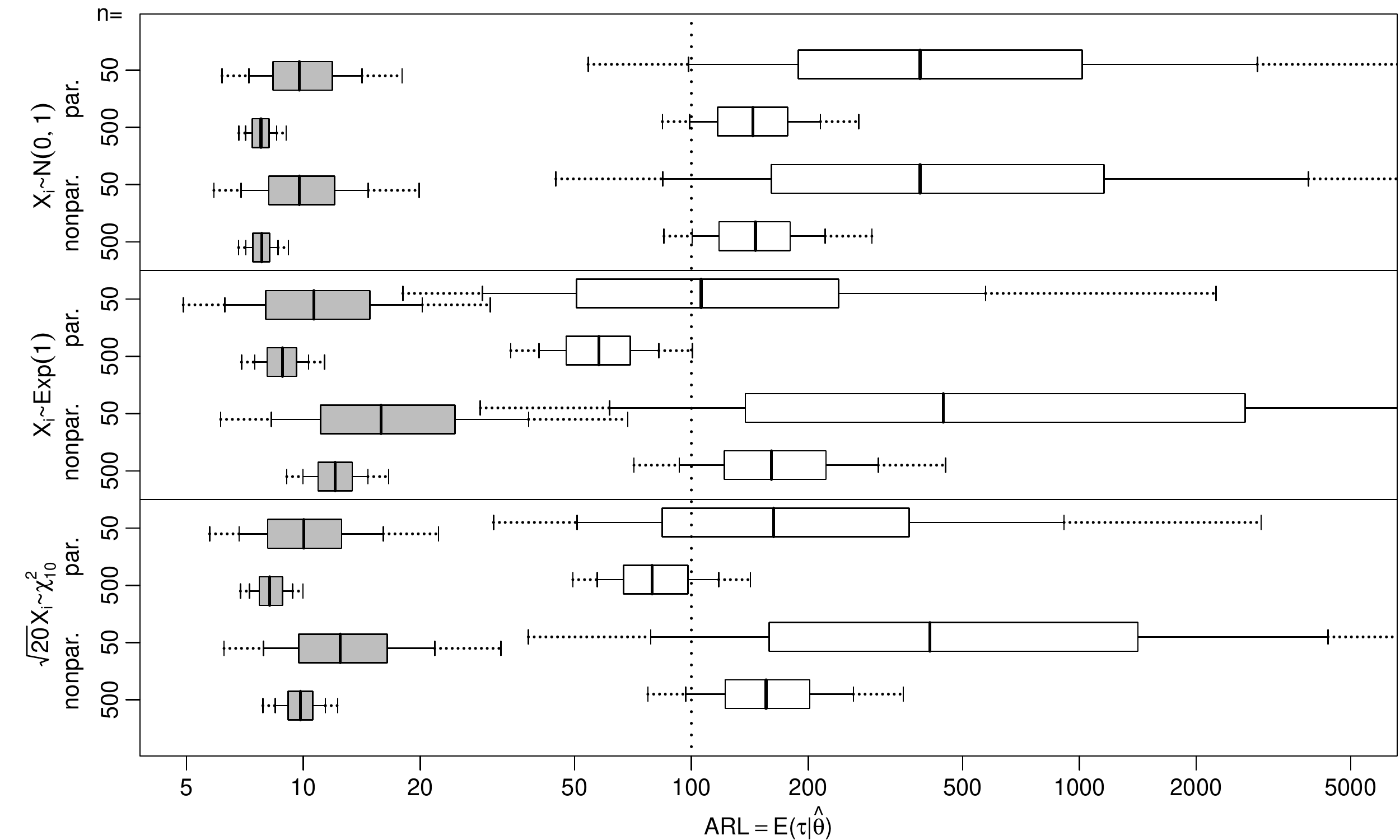}
  \caption{Effects of misspecification. Thresholds are calibrated to
    an in-control ARL of 100 and adjusted to the estimation error with
    a guarantee of 90\%.  A log transform is used in the
    calibration. The white boxplots are in-control, the gray
    boxplots are  out-of-control. The boxplots show
    the 2.5\%, 10\%, 25\%, 50\%, 75\%, 90\% and 97.5\% quantiles.}
  \label{fig:par_nonpar_ARL}
\end{figure}

In the correctly specified model ($X_t\sim N(0,1)$), the performance
of the parametric and the non-parametric chart seems to be almost
identical. The only difference is a slightly worse in-control
performance for the non-parametric chart for $n=50$.

In the misspecified model with $X_t\sim \text{Exponential}(1)$, the
parametric chart does not have the desired in-control
probabilities. The non-parametric chart seems to be doing well, in
particular for $n=500$.  We have a similar results in the other
misspecified model, with $\sqrt{20}X_t\sim \chi^2_{10}$.

\section{Regression models}
\label{sec:regmod}

In many monitoring situations, the units being monitored are heterogeneous,
for instance when monitoring patients at hospitals or bank customers.
To make sensible monitoring systems in such situations, the explainable
part of the heterogeneity should be accounted for by relevant
regression models. The resulting charts are often called risk
adjusted, and an overview of some such charts can be  found in \cite{Grigg2004oor}.

To run risk adjusted charts, the regression model needs to be estimated
based on past data, and this estimation needs to be accounted for.  Our
approach for setting up charts with a guaranteed performance applies
also to risk adjusted charts,
and we will in particular look at linear, logistic and survival
models.

\subsection{Linear models}
\label{subsec:linmod}

Suppose we have independent observations $(Y_1,X_1),$  $(Y_2,X_2)$, $\ldots$,
where $Y_i$ is a response of interest and $X_i$ is a corresponding
vector of covariates, with the first component usually equal to 1.
Let $P$ denote the joint distribution of $(Y_i,X_i)$ and suppose that
in control $\E(Y_i|X_i)=X_i\xi$.  From some observation
$\kappa$ there is a shift in the mean response to
$\E(Y_i|X_i)=\Delta+X_i\xi$ for $i=\kappa,\kappa+1,\dots$.

Monitoring schemes for detecting changes in regression models can
naturally be based on residuals of the model, see for instance
\cite{Brown1975TfT} and \cite{Horvath2004Mci}.
 We can, for instance,
define a CUSUM to monitor changes in the conditional mean of $Y$ by
$$S_t=\max(0, S_{t-1}+Y_t-X_t \xi - \Delta/2), \quad S_0=0, $$
with hitting time $\tau=\inf\{t>0: S_t\geq c\}$ for some threshold
$c>0$.
In a similiar manner we could also set up charts for
monitoring changes in other components of  $\xi$.

The parameter vector $\xi$ is estimated from past in
control data, e.g.\ by the standard least squares estimator.  We
suggest to use a nonparametric version of the general Algorithm
\ref{alg:Bootstrap} with $\hat P$ being the
empirical distribution putting weight $1/n$ on each of the past
observations $(Y_{-n},X_{-n}),\dots,(Y_{-1},X_{-1})$. Resampling is
then equivalent to resampling
$(Y^{\ast}_{-n},X^{\ast}_{-n}),\dots,(Y^{\ast}_{-1},X^{\ast}_{-1})$ by
drawing with replacement from $\hat P$.

The suggested method should  work  even if the linear model is misspecified,
i.e.\  $\E(Y_i|X_i)=X_i\xi$ does not necessarily hold. The nonparametric
bootstrap should take  this into account.

An analogous approach can be used for Shewhart charts. In settings
where it is reasonable to consider the covariate vector to be
non-random one could alternatively use bootstrapping of residuals, see
for example \cite{Freedman1981BRM}.

\subsubsection{Theoretical considerations}

Obtaining precise results is more demanding than in the examples
without covariates in Section~\ref{sec:th:examples}. We only
give an idea of the setup that might be used.

The set of distributions of the observations $D_q$ can be chosen as the
set of cdfs on $\R^{d+1}$ with finite second moments, where $d$ is the dimension of the covariate. The first cdf corresponds to the responses, the others to the covariates. $D_q$ is contained
in the vector space $D=l_{\infty}(\R^{d+1})$, the set of bounded functions $\R^{d+1}\to \R$.
The parameters needed to run the chart are the regression coefficients contained in the set  $\Xi=\R^d$.
These parameters are obtained from the distribution of the observations via  $\xi:D_q\to \Xi$, $F\mapsto
(E(X^TX))^{-1}E(X Y)$ where $(Y,X)\sim F$ where $X$ is considered to
be a row vector.

We conjecture that the conditions of Theorem \ref{th:main} are broadly
satisfied if the cdf  of $Y-X\xi$ is differentiable and if for
the property $q$ we use  hitting probabilities or thresholds to
achieve a given hitting probability.  In particular, it should be
possible to show Hadamard differentiability similarly to Lemma
\ref{le:HaddiffCUSUMhit}: write $q$ as concatenation of two functions
and use the chain rule in Lemma \ref{le:chainrule}.  The first mapping
returns the distribution of the updates of the chart depending on
$F\in D_q$ and $\xi\in \Xi$ via $(F;\xi)\mapsto {\cal L}(Y-X\xi-
\Delta)$, where ${\cal L}$ denotes the law of a random variable.  The
second takes the distribution of the updates and returns the property of
interests. The differentiability of the second map has been shown in
Lemmas \ref{le:Haddiffinversemap} and  \ref{le:diffhitprob}.

\subsubsection{Simulations}
\label{example:CUSUMLinReg}

We illustrate the performance of the bootstrapping scheme using a
CUSUM and the linear in-control model
$Y=X_{1}+X_{2}+X_3+\epsilon$. Let $\epsilon\sim N(0,1)$, $X_{1}\sim
\text{Bernoulli}(0.4)$, $X_2\sim U(0,1)$ and $X_3\sim N(0,1)$, where
$X_1,X_2,X_3$ and $\epsilon$ are all independent. The out-of-control
model is $Y=1+X_{1}+X_{2}+X_3+\epsilon$,
i.e. $\Delta=1$. Figure~\ref{fig:regression_ARL} shows the distribution
of the attained ARL for CUSUMs with thresholds calibrated to give an
in control ARL of 100. We see that the behaviour of the adjusted
versus unadjusted thresholds are very similar to what we observed for
the simpler model in Figure~\ref{fig:adjusted_unadjusted_ARL}. The
coverage probabilities obtained for this regression model, not
reported here, are also very similar to the covarage probabilities
reported in Table~\ref{tab:covprob_simnormal}, though with a tendency
to be slightly worse.
\begin{figure}[tb]
  \centering
  \includegraphics[width=\linewidth]{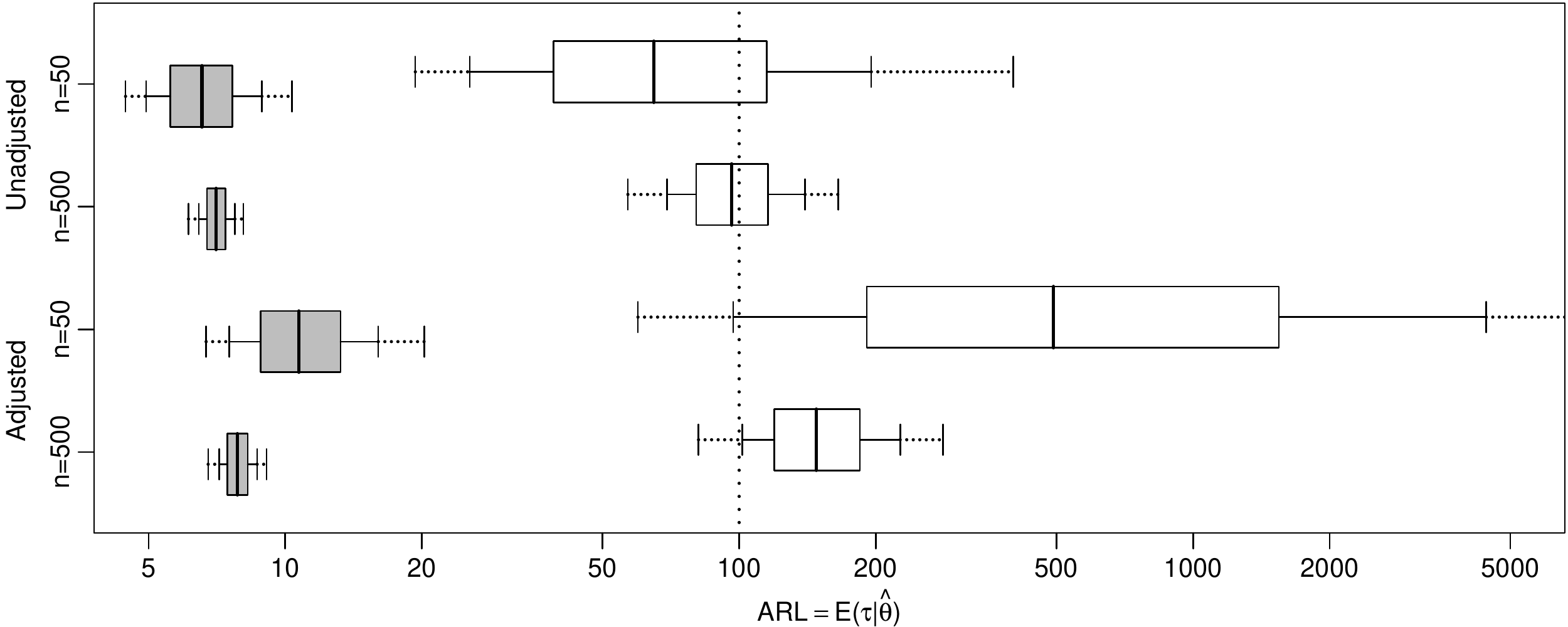}
  \caption{Distribution of the conditional $\ARL$ for CUSUMs in a
    linear regression setup. Thresholds are calibrated to an
    in-control $\ARL$ of 100. A log transform is used in the
    calibration.  The adjusted thresholds have a guarantee of
    90\%. The white boxplots are in control, the gray
    out-of-control. The boxplots show the 2.5\%, 10\%, 25\%, 50\%,
    75\%, 90\% and 97.5\% quantiles.}
  \label{fig:regression_ARL}
\end{figure}

\subsection{Logistic regression}
\label{subsec:logreg}

Control charts, in particular CUSUM charts, based on logistic
regression models are popular for modelling of binary
outcomes in medical contexts. See e.g.\ \cite{Lie1993nsp}, \cite{Steiner2000Msp},
\cite{Grigg2004oor} and \cite{Woodall2006Tuo}.

Suppose we have independent observations $(Y_1,X_1),(Y_2,X_2),\ldots,$
where $Y_i$ is a binary response variable and $X_i$ is a corresponding
vector of covariates. Further, suppose that in control the log odds
ratio is $\logit(\Prob(Y_i=1|X_i))=X_i\xi$, and that from some observation
$\kappa$ there is a shift in the log odds ratio to
$\logit(\Prob(Y_i=1|X_i))=\Delta+X_i\xi$ for $i=\kappa,\kappa+1,\dots$

A CUSUM to monitor changes in the odds ratio can be defined by \citep{Steiner2000Msp}
$$S_t=\max(0, S_{t-1}+R_t), \quad S_0=0, $$
where $R_t$ is the log likelihood ratio between the in-control and out-of-control model for observation $t$. More precisely
$$
\exp(R_t)=\frac{\exp(\Delta+X_t\xi)^{Y_t}/(1+\exp(\Delta+X_t\xi))}{\exp(X_t\xi)^{Y_t}/(1+\exp(X_t\xi))}
=\exp(Y_t\Delta)\frac{1+\exp(X_t\xi)}{1+\exp(\Delta+X_t\xi)}.
$$

The parameter vector $\xi$ is estimated from past in-control
data by e.g.\  the standard maximum likelihood estimator. The same
nonparametric bootstrap approach as described for the linear model in
Section~\ref{subsec:linmod} can now be applied to this CUSUM based on
this logistic regression model. Moreover, this approach would also
apply to control charts based on other generalized linear models, for
instance Poisson regression models for monitoring count data. The only
amendment needed is to replace $R_t$ by the relevant log likelihood
ratio.


We have run simulations, not reported here, based on the same covariate
specifications as in Section~\ref{example:CUSUMLinReg}.  The results
are similar to the results for the linear model of
Section~\ref{example:CUSUMLinReg}.

\subsection{Survival analysis models}

Recently, risk adjusted control charts based on survival models have
started to appear, see
\cite{Biswas2008rCi,Sego2009Rmo,Steiner2009ras,Gandy2010ram}. In none
of these papers any adjustment for estimation error is done, but
\cite{Sego2009Rmo} are  illustrating, by simulations, the impact of
estimation error on the attained average run length  for the
accelerated failure time model based CUSUM studied in their paper.

In the following, we provide a brief simulation example of our
adjustment in a survival setup where we use the methods described in
\cite{Gandy2010ram}.

We observe the survival of individuals over a fixed time interval of
length $n$ (we will use $n=100$ and $n=500$).  Individuals arrive at
times $B_i$ (in our simulation according to a Poisson process with
rate $1$), and survive for $T_i$ time units. Individuals may arrive
before the observation interval, as long as $B_i+T_i$ is after the
start of the observation interval.  Right-censoring, at $C_i$ time
units after arrival, is taking place after a maximum follow-up time of
$t=60$ time units or after the individuals leave the observation
interval.  In the simulation, the true hazard rate of $T_{i}$ is
$h_i(t)= 0.1\exp( X_{1i}+X_{2i})$, where $X_{1i}\sim
\text{Bernoulli}(0.4)$ and $X_{2i}\sim N(0,1)$ are covariates.

Based on the observed data we fit a Cox proportional hazard model
with $X_{1i}$ and $X_{2i}$ as covariates and nonparametric baseline,
giving estimates $\hat \beta$ for the covariate effects and $\hat
\Lambda_0(t)$ for the the integrated baseline.

We use the CUSUM chart described in
\cite{Gandy2010ram} against a proportional alternative with
$\rho=1.25$.  The parameters
needed to run the chart are $\xi=(\beta, \Lambda_0)$ estimated by
$\hat\xi=(\hat \beta, \hat \Lambda_0)$.
To be precise, the chart signals at time
$\tau=\inf\{t>0:S(t)\geq c\}$, where
$S(t)=R(t)-\inf_{s\leq t}R(s)$,
$
                              R (t ) = \log(\rho ) N (t ) - (\rho - 1)
\Lambda(t),
$
 $N (t )$ is the number of events until time $t$ and $\Lambda(t ) =
\sum_{i} \exp( \beta_1X_{i1}+ \beta_2X_{2i})
\Lambda_0(\min((t-B_i)^{+},T_i,C_i))$.

We are interested in finding a threshold that gives a desired
hitting probability, i.e. we use $q=c_{\hit}$.  We compute
$c_{\hit}(P,\xi)$ via simulations (simulate new data from $P$ and run
the chart with $\xi$). We estimate the threshold needed to get a 10\%
false alarm probability in $n$ time units in control, by the 90\% quantile of 500
simulations of the maximum of the chart.

To resample, we resample individuals with replacement. We use 500
bootstrap samples.  Figure \ref{fig:hitprobsurvanal} shows the
distribution of the resulting hitting probabilities based on 500
simulated observation intervals.

\begin{figure}
 \centering
\includegraphics[width=\linewidth]{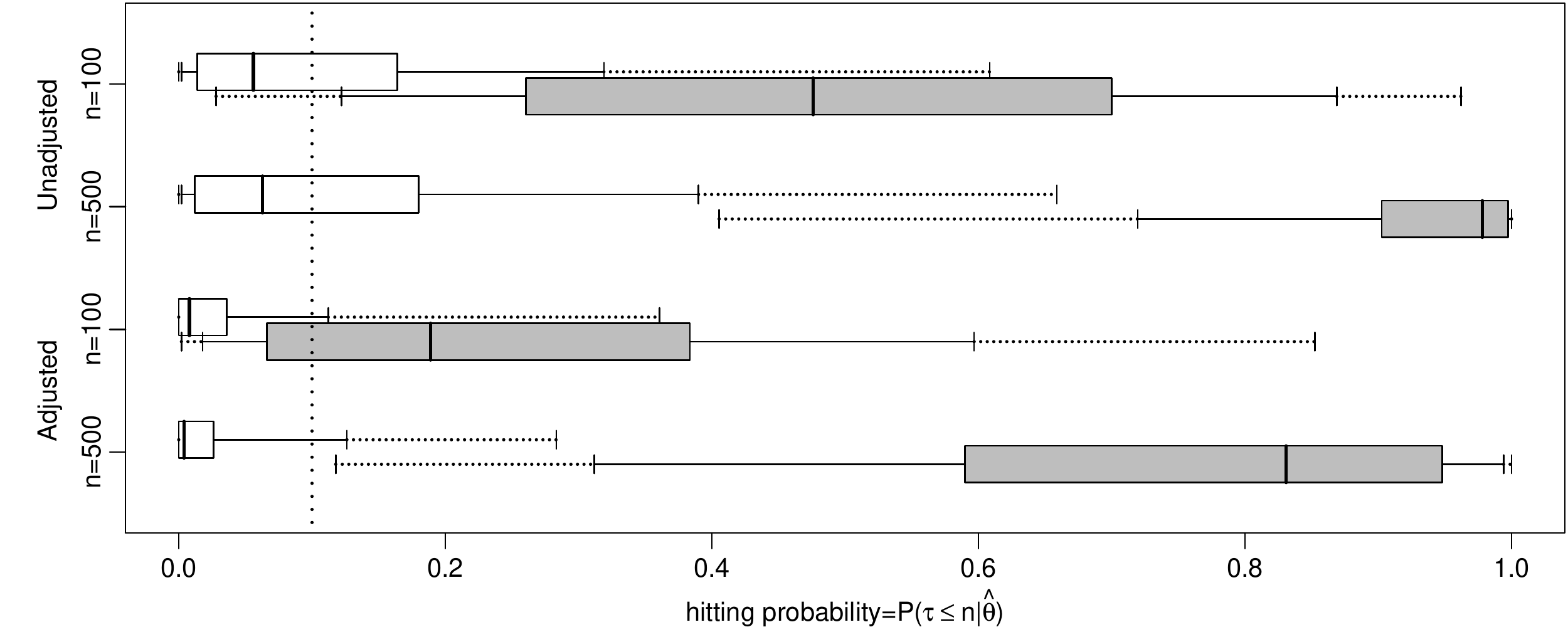}
\caption{Distribution of the conditional hitting probability for
  survival analysis CUSUMs. Thresholds are calibrated to an in-control
  hitting probability of 0.1.  The adjusted thresholds have a
  guarantee of 90\%.  The white boxplots are in control, the gray
  out-of-control. The boxplots show the 2.5\%, 10\%, 25\%, 50\%, 75\%,
  90\% and 97.5\% quantiles. \label{fig:hitprobsurvanal}}
\end{figure}

In control, without the adjustment, the desired false
alarm probability of 0.1 is only reached in roughly 60\% of the cases.  The
bootstrap correction seems to work fine, leading to a false alarm probability
of at most 10\% in roughly 90\% of the cases.  As expected, increasing the
length of the fitting period and the length of time the chart is run
from $n=100$ to $n=500$ results in higher out-of-control hitting probabilities.

If the length of the fitting period and the deployment period of the chart differ then
a somewhat more complicated resampling procedure needs to be used.
For example, one could
 resample arrival times and survival times/covariates separately.
The former could be done by assuming a Poisson process as arrival time and the
latter either by resampling with replacement or by sampling from an estimated
Cox model and an estimated censoring distribution.

\cc{
              In the survival analysis case with a proportional alternative, the chart is  based on
              $$
              R(t) = \log(\rho ) N (t )- (\rho - 1)\hat \Lambda (t ),
              $$
              After the time transformation of $N$ to the standard Poisson process $\tilde N$ this becomes
              \begin{align*}
                \tilde R(t) = R(\Lambda^{-1}(t))=\log(\rho)\tilde
                N(t)-(\rho-1)\hat\Lambda(\Lambda^{-1}(t))
              \end{align*}
              Thus the nice Markov-approximation will not work $\hat
              \Lambda$ and $\Lambda$ will not have independent
              increments. Therefore we needed to simulate.
}

\section{Conclusions and discussion}
\label{sec:conclusion}

We have presented a general approach for handling estimation error in
control charts with estimated parameters and unknown in-control
distributions. Our suggestion is, by bootstrap methods, to tune the
monitoring scheme to guarantee, with high probability, a certain
conditional in-control performance (conditional on the estimated
in-control distribution). If we apply a nonparametric bootstrap, the
approach is  robust against model specification error.

In our opinion, focusing on a guaranteed conditional in-control  performance is
generally more relevant than focusing on some average
performance, as an estimated chart usually is run for some time without
independent reestimation. Our approach can also easily be adapted to
make for instance bias adjustments. Bias adjustments, in
contrast to guaranteed performance, tend to
be substantially influenced by tail behaviour for heavy
tailed distributions which for instance the average run length has.
This implies that the bias adjustments need not be useful in the
majority of cases as the main effect of the adjustment is to adjust
the tail behaviour.

We have in particular demonstrated our approach for various variants
of Shewhart and CUSUM charts, but the general approach will
apply to other charts as well. The method is generally
relevant when the in-control distribution is unknown and the conditions of
Theorem~\ref{th:main} hold. We conjecture that this will be the case
for many of the most commonly used control charts.
\cc{for instance be the case for charts like  EWMA charts \citep{Roberts1959CCT}, general
likelihood ratio based charts \cite{Frisen1991Ops,Frisen2003SSO},  the
Sets method \citep{Chen1978SSC,Grigg2004ARA}}
Numerous extensions of control charts to other settings exist, for example
to other  regression
models, to  autocorrelated data, to multivariate data.
We do  conjecture that our approach will also apply  in
many of these settings.


\small

\appendix

\section{Proof of the main theorem}
\label{sec:proof}

The following extension of the functional delta method will help in the proof of
 Theorem \ref{th:main}.
\begin{lemma}
\label{le:ext_functional_Delta_method}
Suppose that $q:D_q\times \Xi \to E$ is Hadamard differentiable at $P\in D_q$ around $\xi\in \Xi$
tangentially to $D_0\subset D$ and that $\xi:D_q\to \Xi$ is continuous.
Let $\hat P$ be a sequence of random elements in $D_q$ such that
$$
  \sqrt{n}(\hat P-P)\leadsto Z \quad(n\to \infty),
$$
where $Z$ is some tight random element in $D_0$.
Then
$$
\sqrt{n}(q(\hat P;  \xi(\hat P))-q(P;\xi(\hat P))) \leadsto q'(P;\xi(P))Z.
$$
\end{lemma}
\begin{proof}
\cc{We will apply the extended continuous mapping theorem \citep[Th
1.11.1]{Vaart1996WCa} to prove the result.}
Note that
$
\sqrt{n}(q(\hat P; \hat \xi)-q(P;\hat \xi))
=g_n(\sqrt{n}(\hat P - P)),
$
where
$g_n:\tilde D_n\to F$, $g_n(h)=\sqrt{n}[q(P+n^{-\frac{1}{2}}h;\xi(P+n^{-\frac{1}{2}}h))-q(P;\xi(P+n^{-\frac{1}{2}}h))]$
and $\tilde D_n=\{h\in D: P+n^{-\frac{1}{2}}h\in D_q\}$.

Let $h_n$ be a sequence such that $h_n\in \tilde D_n$ and  $h_n\to h$ for some $h\in D_0$.
Let $\xi_n=\xi(P+n^{-\frac{1}{2}}h_n)$.
The continuity of $\xi$ implies $\xi_n\to \xi(P)$.
Thus by the Hadamard differentiability of $q$ we get
  $g_n(h_n)\to q'(P;\xi(P))(h)$.
Using the  extended continuous mapping theorem  \citep[Th 1.11.1]{Vaart1996WCa}
finishes the proof.
\end{proof}

\begin{proof}[Proof of Theorem \ref{th:main}]
Let $\tilde Z_1$ and $\tilde Z_2$ be independent copies of $Z$.
Arguing as in the first part of the  proof of \cite[Theorem 12.1\cc{(p.236/237)}]{Kosorok2008ItE} one can see that unconditionally
$$
\sqrt{n}
\left(
\begin{pmatrix}
  \hat P^{\ast}\\
\hat P\end{pmatrix}
-
\begin{pmatrix}
P\\
P
\end{pmatrix}
\right)
\leadsto
\begin{pmatrix}
  \tilde Z_1+\tilde Z_2\\
\tilde Z_2
\end{pmatrix}.
$$
Applying Lemma
\ref{le:ext_functional_Delta_method} to this with the mappings
 $(x,y;\xi)\mapsto
(q(x;\xi),q(y;\xi),x,y)$  and  $(x,y)\mapsto \xi(x)$ gives
$$
\sqrt{n}
\begin{pmatrix}
  q(\hat P^{\ast};\hat\xi^{\ast})-q(P;\hat \xi^{\ast})\\
  q(\hat P;\hat\xi^{\ast})-q(P;\hat \xi^{\ast})\\
  \hat P^{\ast}-P\\
\hat P-P
\end{pmatrix}
\leadsto
\begin{pmatrix}
  q'(P;\xi)(\tilde Z_1+\tilde Z_2)\\
  q'(P;\xi)(\tilde Z_2)\\
\tilde Z_1+\tilde Z_2\\
\tilde Z_2
\end{pmatrix}.
$$

 After that one can argue
exactly as in the remainder of the  proof of \cite[Theorem 12.1]{Kosorok2008ItE},
p.237, to show that
\begin{equation}
  \label{eq:condbootstrapconv}
 A_n:= \sqrt{n}(q(\hat P^{\ast};\xi(\hat P^{\ast}))-q(\hat P;\xi(\hat P^{\ast}))\condweakconv{\hat P} G \quad(n\to \infty)
\end{equation}
for  $G= q'(P;\xi)Z$. \cc{We need condition e) here.}
Furthermore, Lemma \ref{le:ext_functional_Delta_method} shows that
\begin{equation}
\label{eq:weakconv}
B_n:=\sqrt{n}(q(\hat P; \hat \xi)-q(P;\hat \xi))
\leadsto
 G.
\end{equation}

Similarly to the ideas in \cite[Lemma 23.3]{Vaart1998AS}, one can show
that (\ref{eq:condbootstrapconv}) and (\ref{eq:weakconv}) imply the
correct coverage probabilities.  Indeed, for any subsequence there is a further  subsequence
such that  $A_n\leadsto G$ a.s.\ conditionally on $\hat P$. Using
\cite[Lemma 21.2]{Vaart1998AS}, we get $F^{-1}_{A_n}\leadsto F^{-1}_G$ along this subsequence,
where $F^{-1}$ denotes the quantile function of the random variable in
the subscript, i.e.~$F^{-1}_G(x)=\inf\{t\in \R: P(G\leq t)\geq x\}$.
Thus for any continuity point $\beta$ of $F_G^{-1}$, we get
$
F^{-1}_{A_n}(\beta)\to F^{-1}_G(\beta)\quad a.s.
$ along the subsequence.
Thus overall, we have
$$
F^{-1}_{A_n}(\beta)\stackrel{\Prob}{\to} F^{-1}_G(\beta).
$$
By Slutsky's lemma and (\ref{eq:weakconv}),
$
B_n-F^{-1}_{A_n}(\beta)\leadsto G-F^{-1}_G(\beta).
$
Thus, as $G$ is continuous, \cc{implying $G-F^{-1}_{G}$ is continuous at $0$}
\begin{equation}
  \label{eq:covprobgeneral}
\Prob(B_n\leq F^{-1}_{A_n}(\beta))\to \Prob(G\leq F^{-1}_G(\beta))=\beta.
\end{equation}
This holds for all $\beta$, because there are at most countably many
points $\beta$ at which $F^{-1}_G$ is not continuous, because both the
left- and the right-hand side of (\ref{eq:covprobgeneral}) are
monotone in $\beta$, and because the right-hand side is continuous.
As $p^{\ast}_{\alpha}=F^{-1}_{A_n}(\alpha)/\sqrt{n}$,
(\ref{eq:covprobgeneral}) implies
$$
\Prob(q(P;\hat \xi)<  q(\hat P; \hat \xi)-p_{\alpha}^{\ast})=1-\Prob(q(\hat P; \hat \xi)-q(P;\hat \xi)\leq p_{\alpha}^{\ast})
=1-\Prob(B_n\leq F^{-1}_{A_n}(\alpha))\to 1-\alpha.
$$
\cc{need condition f) for the left hand side to make sense}%
\end{proof}

\section{Proofs for Hadamard differentiability}

The main goal of this section is to prove Lemmas
\ref{le:HaddiffCUSUMhit} and \ref{le:HaddiffCUSUMhitNORMAL}. Before we
do this in Appendix \ref{sec:haddifhitprobex}, we first show a chain
rule in Appendix \ref{sec:chain-rule}, then  a lemma about
the uniform Hadamard differentiability of inverse maps in Appendix
\ref{sec:unif-hadam-diff}. After that we show general differentiability
of hitting probability in CUSUM charts with respect to the updating
distribution in Appendix \ref{sec:diff-hitt-prob}. The results in
\ref{sec:chain-rule}-\ref{sec:diff-hitt-prob} may also be useful in
other situations.

\subsection{Chain rule}
\label{sec:chain-rule}
In this section we present a chain rule for Hadamard differentiable functions.
For this we need the following stronger version of  Hadamard differentiability.

\begin{definition}
  Let $D,E$ be metric spaces. \cc{i.e. vector spaces with a metric} A
  function $\phi:D_{\phi}\subset D \to E$ is called \emph{uniformly
    Hadamard differentiable at $\theta\in D_{\theta}$ along $d:D\times
    D\to\R$
    \cc{we need this for the inverse mapping - in many cases
      we can set $d=0$}
     tangentially to $D_0\subset D$} if there
  exists a linear  map $\phi'_{\theta}:D_0\to E$ such that
$$
\frac{\phi(\theta_n+t_nh_n)-\phi(\theta_n)}{t_n}
\to \phi'_{\theta}(h)
$$
for all $\theta_n\to \theta$ with $d(\theta_n,\theta)\to 0$\cc{for inverse mappings - need the derivative of the functions to converge}, $t_n\to
0$ and all converging sequences $(h_n)$ with $h_n\to h\in D_0$ and
$\theta_n+t_nh_n\in D_{\phi}$.
  \end{definition}
\cc{similar to \cite[p.375]{Vaart1996WCa}}

\begin{lemma}[Chain rule]
\label{le:chainrule}
Let $D$, $E$, $F$ be metric spaces and let $H$ be a non-empty set.
Let $\{f_\xi:D_f \to E: \xi\in \Xi\}$ be a family of functions that is
Hadamard differentiable at $\theta\in D_f$ around $\xi\in \Xi$
tangentially to $D_0\subset D$.
  Let $\phi:E_{\phi}\to F$ be uniformly
Hadamard differentiable at $f_{\xi}(\theta)$ along $d: E_{\phi}\times
E_{\phi}\to \R$ tangentially to $f_{\xi}'(\theta;\xi)(D_0)$.
Furthermore, suppose that $\xi_n\to \xi$ implies $d(f(\theta;\xi_n),
f(\theta;\xi))\to 0$.
  Then $\{\phi\circ f_\xi:D_f \to F: \xi\in
\Xi\}$  is Hadamard differentiable at $\theta$ around $\xi\in \Xi$
tangentially to $D_0$.
\end{lemma}

\begin{proof}
Let  $(\xi_n)\subset \Xi$, $(t_n)\subset \R$, $(h_n)\subset D$
satisfying $\theta+t_nh_n\in D_f \,\forall n$ and $\xi_n\to \xi$, $t_n\to 0$, $h_n\to h\in D_0$ as $n\to \infty$.
Let   $k_n=\frac{f_{\xi_n}(\theta+t_nh_n)-f_{\xi_n}(\theta)}{t_n}$. Hadamard differentiability of $f$ implies
$k_n\to q'_{\xi}(h)$.
Then by uniform Hadamard differentiability of $\phi$,
\[
\frac{\phi(f_{\xi_n}(\theta+t_nh_n))-
\phi(f_{\xi_n}(\theta))}{t_n}=
\frac{\phi(f_{\xi_n}(\theta)+t_nk_n)-
\phi(f_{\xi_n}(\theta))}{t_n}
\to \phi'_{f_{\xi}(\theta)}(q'_{\xi}(\theta)(h)).
\]\qedhere
\end{proof}

\subsection{Uniform Hadamard differentiability of the inverse map}
\label{sec:unif-hadam-diff}

Let $D_{\phi}$ be the set of non-decreasing functions in $D[u,v]$, for some $ -\infty<u< v<\infty$,
 that cross $\beta\in \R$, i.e.
$$D_{\phi}=\{F\in D[u,v]: F \text{ non-decreasing}, \exists x\in (u,v]: F(x-)\leq \beta\leq F(x)\}.$$

Suppose that $F\in D_{\phi}$ and $G\in D_{\phi}$ are differentiable on
$[u,v]$ with derivatives $f$ and $g$.  Let $d(F,G)=\sup_{x\in
  [u,v]}|f(u)-g(u)|$.  If either $F$ or $G$ are not differentiable on
$[u,v]$ then we set $d(F,G)=\infty$.

Let $\phi:D_{\phi}\to \R, \phi(F)=\inf\{x:F(x)\geq \beta\}$, the first point at which the function crosses the threshold.

\begin{lemma}
\label{le:Haddiffinversemap}
  Let $\theta\in D_{\phi}$ such that $\theta$ is differentiable on
  $[u,v]$ with continuous bounded positive derivative. Then $\phi$ is
  uniformly Hadamard differentiable at $\theta$ along $d$ tangentially to
  $C[u,v]$.
\end{lemma}

\begin{proof}
  \cc{This proof extends the proof given in
    \cite[p.385]{Vaart1996WCa}. }

Let $\xi=\phi(\theta)$.
Let $(h_n)\subset D[u,v]$ such that $h_n\to h\in C[u,v]$. Let $(t_n)\subset [0,\infty)$ such that $t_n\to 0$.
Let $(\theta_n)\subset D_{\phi}$ such that $\theta_n\to \theta$ \cc{wrt the uniform norm} and $d(\theta_n,\theta)\to 0$. \cc{The latter means   $\theta_n'\to\theta'$ uniformly.}
Let $\xi_{ n }= \phi(\theta_n+t_nh_n)$. By the definition of $\phi$, we have
\begin{equation}
  \label{eq:invproof_definv}
(\theta_n+t_nh_n)(\xi_n-\epsilon_n)\leq \beta \leq (\theta_n+t_nh_n)(\xi_n).
\end{equation}
for every $\epsilon_n>0$. Let  $(\epsilon_n)$ be positive and such that
$\epsilon_n=o(t_n)$. \cc{Could choose $\epsilon_n=t_n^2$.}

First, we show  $\xi_n\to \xi$. The
sequence $(h_n)$ is uniformly bounded because $h_n\to h$ and because
$h$ is bounded.  Thus,
$$
\theta_n(\xi_n-\epsilon_n)+O(t_n)\leq \beta \leq \theta_n(\xi_n)+O(t_n).
$$
As $t_n\to 0$ and $\theta_n\to \theta$\cc{which is uniform convergence},
$$
\theta(\xi_n-\epsilon_n)+o(1)\leq \beta \leq \theta(\xi_n)+o(1).
$$
For every $\delta>0$, the function $\theta$ is  bounded away from $\beta$ outside $(\xi-\delta, \xi+\delta)$. Furthermore, $\theta$ is strictly increasing.
Thus, to satisfy the previous display we must have eventually  $\xi_n\geq \xi-\delta$ and
$\xi_n-\epsilon_n\leq \xi+\delta$, which implies $\xi_n\to \xi$.

 Let $\tilde \xi_n=\phi(\theta_n)$. Using the mean value theorem  in (\ref{eq:invproof_definv}) gives
\begin{equation*}
\theta_n(\tilde \xi_n)+(\xi_n-\epsilon_n-\tilde \xi_n)\theta_n'(\rho_{1n})   +t_nh_n(\xi_n-\epsilon_n)\leq
\beta \leq
\theta_n(\tilde \xi_n)+(\xi_n-\tilde \xi_n)\theta_n'(\rho_{2n})   +t_nh_n(\xi_n)
\end{equation*}
for some $\rho_{1n}$ between  $\xi_n-\epsilon_n$ and $\tilde \xi_n$ and for some $\rho_{2n}$ between $\xi_n$ and $\tilde \xi_n$. \cc{\cite[p.385]{Vaart1996WCa} use a Taylor approximation at this point. However, for this they would need the second derivative of to exist - which they and we do not have. I suppose that \cite{Vaart1996WCa} really want to use the definition of differentiability at this point.}.
Rewriting this using  $\theta_n(\tilde \xi_n)=\beta$ gives
\begin{equation*}
(\xi_n-\tilde \xi_n)\theta_n'(\rho_{1n})   +t_nh_n(\xi_n-\epsilon_n)
-\epsilon_n\theta_n'(\rho_{1n})\leq
0 \leq
(\xi_n-\tilde \xi_n)\theta_n'(\rho_{2n})   +t_nh_n(\xi_n).
\end{equation*}
By the uniform convergence of $h_n$ and the continuity of $h$, we have
$h_n(\xi_n-\epsilon_n)\to h(\xi)$ and  $h_n(\xi_n)\to h(\xi)$.
Using this, the fact that
we have chosen $\epsilon_n$ such that $\epsilon_n=o(t_n)$ and that   $\theta_n'$ is uniformly bounded \cc{this follows because $\theta_n'\to \theta'$ uniformly and because $\theta$ is continuous and hence bounded},
we get
 \begin{equation*}
(\xi_n-\tilde \xi_n)\theta_n'(\rho_{1n})
-o(t_n)\leq
-t_nh(\xi) \leq
(\xi_n-\tilde \xi_n)\theta_n'(\rho_{2n})   +o(t_n).
\end{equation*}
Hence, \cc{using that $\theta'_n$ is uniformly bounded from below - which follows through uniform convergence from the corresponding property of $\theta'$}
 \begin{equation*}
-\frac{h(\xi)}{\theta_n'(\rho_{2n})} -o(1)\leq \frac{\xi_n-\tilde \xi_n}{t_n}
\leq
-\frac{h(\xi)}{\theta_n'(\rho_{1n})}+o(1).
\end{equation*}
 We have already shown
$\xi_n\to \xi$ which implies $\rho_{1n}\to \xi$ and  $\rho_{2n}\to \xi$. Using the assumptions
that $\theta_n'\to \theta$ uniformly and that $\theta'$ is continuous shows
 $\theta_n'(\rho_{1n})\to \theta'(\xi)$
and  $\theta_n'(\rho_{2n})\to \theta'(\xi)$, which finishes the proof.
\cc{The alternative
  approach of using the theorem about Hadamard differentiability in
  vdVaart and adapting the lemma about sufficient conditions of
  uniform Hadamard differentiability does not seem to work. Mainly because
   it will require differentiability at $\theta_n+t h_n$, which we
  may not have}%
\end{proof}

\subsection{Differentiability of the  hitting probability with respect to  the updating distribution}
\label{sec:diff-hitt-prob}

\cc{The next section essentially looks at the middle mapping of a chain of mappings that we are interested in.
$(P,\xi)\mapsto F \mapsto (\text{hit}_c)_{c>0} \mapsto c_{\text{hit}}$. We need to show differentiability of the first step (which will be different for the parametric and the nonparametric and the regression approach.
}

We are interested in hitting probabilities for CUSUM charts within the
first $T\in \N$ steps. We will show that the mapping from the
distribution of the updates $Y_i$ to the hitting probabilities is
uniformly Hadamard differentiable. The $Y_i$ are the adjusted observations, e.g.~
in the notation of Section \ref{sec:theor:ex:CUSUM:nonpar} they are
$Y_i=\frac{X_i-\xi_1-\Delta/2}{\xi_2}$, where $X_i$ is the observed value.

Let $D_\phi$ be the set of cumulative distribution functions on $\R$,
considered as a subset of $D=l_{\infty}(\R)$ equipped with the uniform
norm.  Consider the mapping
$$
\phi:D_\phi\to l_{\infty}(\R),
q(F)(c)=P(\text{hit threshold c within T})= \int g(y,c)d F(y_1) \dots dF(y_T),
$$
where $g(y,c)=1\left(m(y)\geq c\right)$ with $1(\cdot)$ being the
indicator function and $m(y)$ is as defined in Section
\ref{sec:theor:ex:CUSUM:nonpar}.

\begin{lemma}
\label{le:diffhitprob}
$\phi$ is uniformly Hadamard differentiable tangentially to
$D_0=\{H\in C(\R):\lim_{t\to \infty}H(t)=\lim_{t\to-\infty}H(t)=0\}$
with derivative
$$
\phi'(F)(H)(c)
=\sum_{i=1}^T\int g(y,c)\left(\prod_{j\neq i} dF(y_j)\right)dH(y_i).
$$
\end{lemma}

Since $H$ may be of infinite variation, the above integral is  defined by partial integration,
i.e.\
$$
\phi'(F)(H)(c)
=-\sum_{i=1}^T\int H(y_i)d\left(\int g(y,c)\left(\prod_{j\neq i} dF(y_j)\right)\right).
$$
\begin{proof}
Suppose $F_n\to F$, $t_{n}\to 0$, $H_n\to H\in D_0$ such that $F_n+t_nH_n\in D_\phi$ for all $n$.
The difference quotient can be written as
\begin{equation}
  \label{eq:hitprobnonpar_diffquot}
  \begin{split}
\frac{\phi(F_n+t_nH_n)(c)-\phi(F_n)(c)}{t_n}=&
\sum_{i=1}^T\int g(y,c)\left(\prod_{j\neq i} dF_n(y_j)\right)dH_n(y_i)
+\!\!\!\!\!\! \sum_{\substack{I\subset \{1,\dots,T\}\\|I|\geq 2}}\!\!\!\!\!\!t_n^{|I|-1}\!A_I,
  \end{split}
\end{equation}
where $A_I=\int \!g(y,c) \left(\prod_{i\notin I}dF_n(y_i)\right) \left(\prod_{i\in I}dH_n(y_i)\right).$
We first show that the second terms converges uniformly in $c$ to $0$.
Partial integration (applied several times) gives that for $I\subset \{1,\dots,T\},|I|\geq 2$,
\begin{equation}
  \label{eq:hitprobnonpar_diffquot_partint}
A_I
=
(-1)^{|I|}\int \left(\prod_{i\in I}H_n(y_i)\right)dB_I(y_I),
\end{equation}
where $B_I=\int g(y,c)  \prod_{i\notin I}dF_n(y_i)$.
\cc{There is no remainder term here, as $F_n+t_nH_n$ and $F_n$ are cdfs on $\R$ and thus $H_n(-\infty)=H(\infty)=0$.}
 $g(y,c)$ is monotonically increasing in $y$ thus $B_I$
is increasing in $y_{I}=(y_i)_{i\in I}$. Thus the total variation of $B_I$ is bounded by 1.
\cc{$\sup _c|\sup_{y_{I}}\int g(y,c)\left(\prod_{i\notin I} dF_n(y_i)\right)-
\inf_{y_{I}}\int g(y,c)\left(\prod_{i\notin I} dF_n(y_i)\right)|\leq 1.$}
Hence, using (\ref{eq:hitprobnonpar_diffquot_partint}),
$$
t_n^{|I|-1}A_I
\leq t_n^{|I|-1}\left(\sup_{z\in \R}|H_n(z)|\right)^{|I|},$$
which converges to $0$ uniformly in $c$.
Thus the second term of (\ref{eq:hitprobnonpar_diffquot}) converges to $0$ uniformly in $c$.

Next, we show that first term on the right hand side of  (\ref{eq:hitprobnonpar_diffquot}), henceforth denoted by $\zeta$, converges uniformly in $c$ to $\phi'(F)(H)$.
Consider the decomposition
\begin{equation}
  \label{eq:decomp_deriv}
  \begin{split}
\zeta
- \phi'(F)(H)
=C_n+\sum_{i=1}^T\int g(y,c)\left((\prod_{j\neq i} dF_n(y_j))-\prod_{j\neq i}dF(y_j)\right)dH(y_i),
  \end{split}
\end{equation}
where $C_n=\sum_{i=1}^T\int g(y,c)\left(\prod_{j\neq i} dF_n(y_j)\right)\left(dH_n(y_i)-dH(y_i)\right)$.
  As mentioned above,
$H$ might be of infinite variation,  thus $C_n$ is
defined via partial integration, i.e.
\begin{align*}
C_n=-\sum_{i=1}^T
   \int (H_n(y_i)-H(y_i))d\left(\int g(y,c)  \prod_{j\neq i}dF_n(y_j)\right).
\end{align*}
\cc{Again, $H_n(-\infty)=H(-\infty)=H(\infty)=H_n(\infty)=0$, thus the other terms disappear.}
As above, the total variation of the integrator is bounded by 1, thus
$
\left| C_n \right|\leq
T \|H_n-H\|,
$
which converges to $0$ as $n\to \infty$.

We can rewrite the second term of  (\ref{eq:decomp_deriv})
as
$$\sum_{i=1}^T\sum_{k=1,k\neq i}^T\int D_{ik}\left((\prod_{j\neq i,j< k} dF_n(y_j))\prod_{j\neq i,j> k} dF(y_j)\right).
$$
where $D_{ik}=\int g(y,c)dH(y_i)(dF_n(y_k)-dF(y_k))$. Using partial integration,
\begin{align*}
D_{ik}
=&
-\!\!
\int\! \!H(y_i)dg(y_{-i},dy_i,c)(dF_n(y_k)-dF(y_k))
=
\!\!\int\!\! H(y_i)(F_n(y_k)-F(y_k))dg(y_{-i,-k},dy_i,dy_k,c),
\end{align*}
where negative subscripts  denote removal of the corresponding component of the vector
(e.g.\ $y_{-i}$ is the vector $y$ with the $i$th component removed).
Since $g$ is of bounded variation with respect to $y_i$ and $y_k$ independent of $c$ and $y_{-i,-k}$,
we can bound this uniformly above by
$
 K \sup_z |H(z)|\sup_z |F_{n}(z)-F(z)|
$ for some fixed $K>0$.
Thus, since the variation of $F_{n}-F$ is bounded by $2$,
the second term of   (\ref{eq:decomp_deriv}) converges to $0$ uniformly in $c$.
\end{proof}

The following lemma is needed to use the result about the inverse mapping, see Lemma \ref{le:Haddiffinversemap}.

\begin{lemma}
  \label{le:diffcusumthreshold}
  Let $F$ be a cumulative distribution function with continuous
  bounded positive derivative $f$.    Let
  $Y=(Y_1,\dots,Y_T)$ where $Y_1,\dots,Y_T\sim F$ independently. Then
  the following holds.
\begin{itemize}%
\setlength{\itemsep}{0pt}%
\setlength{\parskip}{0pt}%
\item[a)] $c \mapsto \Prob(m(Y)\leq c)$ is continuously differentiable for $c>0$ (call this derivative $g$).
\item[b)] $g$ is  bounded away from $0$ (at least on some compact sets),
\item[c)] Let $f_n$ be densities that converge uniformly to $f$.
       Let  $Y^{(n)}=(Y^{(n)}_1,\dots,Y^{(n)}_{T})\sim f_{n}$  and denote  the derivative  of  $c \mapsto \Prob(m(Y^{(n)})\leq c)$ for $c>0$ by $g^{(n)}$.
 Then $g^{(n)}$ converges uniformly to $g$ on any compact set $K\subset (0,\infty)$.
\end{itemize}
\end{lemma}

\begin{proof}
For $0 \leq i \leq T$, let $A_{i}=\{y\in \R^T: R_i(y)>R_{\nu}(y) \forall \nu \neq i\}$.
$A_i$ are  disjoint sets with $\Prob(Y\in \bigcup_iA_i)=1$.
Thus,
$\Prob(m(Y)\leq c)=
\sum_i \Prob(m(Y)\leq c, Y\in A_i)
$
and
$g(c)=\sum_{i=1}^T g_i(c)$,  where
\begin{align*}
&g_i(c)=\frac{\partial}{\partial c} \Prob(m(Y)\leq c, Y\in A_i)=
\frac{\partial}{\partial c} \Prob(R_i(Y)\leq c, Y\in A_i)\\
&=
\frac{\partial}{\partial c} \Prob(Y_i\leq c- R_{i-1}(Y), Y\in A_i)
=
 \int \frac{\partial}{\partial c}\int^{c- R_{i-1}(y)} 1(y\in A_i)f(y_i)dy_i \prod_{\nu\neq i}f(y_{\nu})dy_{-i}\\
&=
 \int  1((y_1,\dots,y_{i-1},{c- R_{i-1}(y)}, y_{i+1}, \dots,y_T)\in A_i)
 f({c- R_{i-1}(y)})\prod_{\nu\neq i}f(y_{\nu})dy_{-i}.
\end{align*}
\cc{note that $R_{i-1}(y)$ does not depend on $y_{i}$.}
The continuity of $g$ follows because because of the  dominated convergence theorem.
\cc{The dominating function is
$
 \|f\|_{\infty} \prod_{\nu\neq i}f(y_{\nu})
$.  This is integrable (the continuity of $f$ and its integrability ensure that $\|f\|_{\infty}$ is finite.}
This shows a).
Statement b) follows from $g$ being positive and its continuity.
For c), use a telescoping sum to go from the product of  $f$s to the product of $f_n$s. Then use the dominated convergence theorem.
\end{proof}

\subsection{Hadamard differentiability of hitting probability in simple examples}
\label{sec:haddifhitprobex}
\begin{proof}[Lemma \ref{le:HaddiffCUSUMhit}]
  $\hit$ can be written as $\phi\circ g$, where $\phi$ is as in Appendix \ref{sec:diff-hitt-prob} and
 $g:D_q\to D_q, g(P;\xi)=(x\mapsto
  P(x\xi_2+\xi_1+\Delta/2)$. \cc{$\Prob(Y\leq x)=\Prob(\frac{X-\xi_1-\Delta/2}{\xi_2}\leq x)
=\Prob(X\leq \xi_1+\Delta/2+\xi_2 x).$}

  We will show that $g$ is Hadamard differentiable at $P$ around $\xi$
  tangentially to $D_0$. Clearly, $g$ is linear in $P$. Thus for
  $t_n\to 0$, $h_n\to h\in D_0$, $\xi_n\to \xi$,
\begin{align*}
\frac{g(P+t_nh_n;\xi_n)-g(P;\xi_n)}{t_n}-g(h;\xi)=&
  g(h_n;\xi_n)-g(h;\xi)\\=&
  (g(h_n;\xi_n)-g(h;\xi_n))+(g(h;\xi_n)-g(h;\xi)).
\end{align*}
The first term
  converges uniformly to $0$ as $h_n\to h$. The second term converges
  to $0$ as $h\in D_0$ implies that $h$ is uniformly continuous.

  Lemma \ref{le:diffhitprob} allows us to use the chain rule in Lemma
  \ref{le:chainrule}, to show the differentiability of $\hit$.

The differentiability of   $c_{\hit}$ can be seen as follows:
$\xi_n\to \xi$ implies that the derivative of $g(P;\xi_n)$ converges uniformly to
the derivative of $g(P;\xi)$. As
$g(P;\xi_n)'(x)=f(x\xi_{2n}+\xi_{1n}+\Delta/2)\xi_{2n}$ this is implied by the uniform continuity of $f$.
Thus, by  Lemma \ref{le:diffcusumthreshold},
the derivative of $\hit(P;\xi_n)$ converges uniformly to the derivative of
$\hit(P;\xi)$.
Thus, the result follows using the chain rule (Lemma
\ref{le:chainrule}), the differentiability of $\hit$, and the
differentiability of the inverse (Lemma \ref{le:Haddiffinversemap}).
\end{proof}

\begin{proof}[Lemma \ref{le:HaddiffCUSUMhitNORMAL}]
Let $g:(\R\times (0,\infty))^2\to D_q$,
 $g(\mu,\sigma,\xi)=(x\mapsto
\Phi(\frac{\xi_1+\Delta/2+\xi_2 x-\mu}{\sigma}))$.
\cc{$\Prob(Y\leq x)=\Prob(\frac{X-\xi_1-\Delta/2}{\xi_2}\leq x)
=\Prob(X\leq \xi_1+\Delta/2+\xi_2 x)=
\Phi(\frac{\xi_1+\Delta/2+\xi_2 x-\mu}{\sigma}).$}
Then, as in the proof of   Lemma \ref{le:HaddiffCUSUMhit},
 $\hit^N=\phi\circ g$.
The proof can be completed with similar steps.
\end{proof}

\setlength{\bibsep}{0cm}
\begin{articlestyle}
\end{articlestyle}

\end{document}